\documentclass[11pt]{article}
\usepackage[T1]{fontenc}
\usepackage[utf8]{inputenc}
\usepackage{amsmath,braket, amsthm}
\usepackage{graphicx}
\usepackage{subcaption}
\usepackage[top=2cm,
            margin=0.9in]{geometry}
\usepackage[colorlinks=true,
            urlcolor=teal,
            linkcolor=teal,
            citecolor=teal,]{hyperref}
\usepackage{cleveref}
\usepackage{xcolor}
\usepackage{dutchcal}
\usepackage{authblk}
\usepackage{comment}
\usepackage[toc,page]{appendix}
\usepackage{algorithm, setspace}
\usepackage{algpseudocode}
\usepackage{kpfonts}
\usepackage{microtype}
\usepackage{thmtools}
\usepackage{thm-restate}

\usepackage[
backend=biber,
style=alphabetic,
sorting=ynt,
maxbibnames=99
]{biblatex}
\newcommand{\semigeq}{\succeq}

\renewcommand\footnotemark{}

\newtheorem{theorem}{Theorem}
\newtheorem*{theorem*}{Theorem}
\newtheorem{proposition}[theorem]{Proposition}

\newtheorem{lemma}[theorem]{Lemma}

\newtheorem{corollary}[theorem]{Corollary}

\newtheorem{definition}[theorem]{Definition}

\newcommand{\cqvc}{EVC}
\newcommand{\ketbra}[2]{|#1\rangle \langle #2|}
\newcommand{\tr}{\mathrm{Tr}}

\newcommand{\knote}[1]{\textcolor{red}{({\bf Kevin:} #1)}}
\newcommand{\cnote}[1]{\textcolor{blue}{({\bf Chaithanya:} #1)}}
\newcommand{\onote}[1]{\textcolor{orange}{(Ojas: #1)}}

\renewcommand{\knote}[1]{}
\renewcommand{\cnote}[1]{}
\renewcommand{\onote}[1]{}

\title{Constrained local Hamiltonians: quantum generalizations of Vertex Cover}

\author[1]{Ojas Parekh\thanks{\text{\{chaithanyarss\}@unm.edu, \{kevthom, odparek\}@sandia.gov}}}
\author[2]{Chaithanya Rayudu}
\author[1]{Kevin Thompson} 
\affil[1]{Quantum Algorithms and Applications Collaboratory, Sandia National Laboratories, Albuquerque, NM, USA}
\affil[2]{Department of Physics and Astronomy and Center for Quantum Information and Control, University of New Mexico, Albuquerque, NM, USA}

\date{}

\begin{document}

\maketitle

\begin{abstract}
Recent successes in producing rigorous approximation algorithms for local Hamiltonian problems such as Quantum Max Cut have exploited connections to unconstrained classical discrete optimization problems.  We initiate the study of approximation algorithms for \emph{constrained} local Hamiltonian problems, using the well-studied classical Vertex Cover problem as inspiration. We consider natural quantum generalizations of Vertex Cover, and one of them, called \emph{Transverse Vertex Cover} (TVC), is equivalent to the PXP model with additional 1-local Pauli-$Z$ terms.
We show TVC is StoqMA-hard and develop an approximation algorithm for it based on a quantum generalization of the classical local ratio method.  This results in a simple linear-time classical approximation algorithm that does not depend on solving a convex relaxation. We also demonstrate our quantum local ratio method on a traditional unconstrained quantum local Hamiltonian version of Vertex Cover which is equivalent to the anti-ferromagnetic transverse field Ising model.
\end{abstract}

\newpage
\tableofcontents
\newpage

\section{Introduction}

The $k$-Local Hamiltonian (LH) problem, the problem of finding ground energy of $k$-local Hamiltonians, forges a fundamental connection between physical models of matter and the theory of quantum complexity and algorithms. LH is a natural quantum generalization of classical constraint satisfaction problems (CSPs)~\cite{kempe2006complexity}. Casting CSPs as discrete optimization problems has inspired rich discoveries, notably for approximation algorithms and hardness of approximation~\cite{haastad2001some, khot2007optimal}.  However, while CSPs are quintessential unconstrained discrete optimization problems, the vast majority of discrete optimization focuses on \emph{constrained} problems.  While these constrained optimization problems in NP reduce to unconstrained NP-complete problems such as MaxCut, such reductions do not generally preserve approximability.  An emerging body of work has established techniques for generalizing approximation algorithms for CSPs to the quantum local Hamiltonian setting.  Yet, quantum counterparts of constrained discrete optimization problems have received virtually no attention~\cite{gharibian2022improved, gharibian2022quantum}. 

Our primary motivation is a better understanding of the approximation and hardness of local Hamiltonian problems, including introducing new frameworks for designing approximation algorithms for local Hamiltonian problems.  While the resolution of the quantum PCP conjecture remains elusive, a better understanding of it might be achieved via QMA-hard instances of local Hamiltonians for which a tight approximability result can be established, with respect to say NP-hardness or Unique-Games-hardness rather than QMA-hardness.  Quantum Max Cut is such a candidate; however, closing the approximability gap remains a challenge, with the currently best-known approximation ratio of 0.595..~\cite{lee2024improved}, and a conjectured Unique-Games-hardness threshold of 0.956... ~\cite{hwang2023unique}.  

Looking at classical problems for inspiration, one of the earliest examples of a tight approximation result based on the Unique Games conjecture is a 2-approximation for Vertex Cover~\cite{khot2008vertex}.  Given a graph, Vertex Cover seeks to find a minimum-size (or weight) set, $C$ of vertices so that each edge contains at least one endpoint in $C$.  Vertex Cover is a prime example of a simple covering problem, in much the same way that Max Cut is an archetypical CSP, and the problem and its generalizations have been widely studied~\cite{Has06, PANDEY2018121}.  Vertex Cover offers attractive features as a model for understanding approximability: (i) several simple optimal approximation algorithms are known for it, some of which are elementary linear-time algorithms that do not require solving semidefinite or linear programs, and (ii) optimal hardness results are easier to prove than for alternatives such as Max Cut~\cite{khot2008vertex, khot2007optimal}.  Given the relative immaturity of quantum approximation algorithms, a natural question is whether quantum generalizations of Vertex Cover might retain these features.

\subsection{Contributions}
In this work, we initiate the study of constrained quantum local Hamiltonian problems from an approximation algorithms perspective. Our broader goal is to understand how interesting constrained local Hamiltonian problems can arise from well-studied classical constrained discrete optimization problems.  The latter has enjoyed a rich history spanning decades and inspired a variety of novel algorithmic techniques.  We hope for the  types of connections we introduce to open doors to new kinds of quantum algorithms inspired by these classical techniques.  While is not clear that quantum local Hamiltonian problems obtained in this way should be intrinsically interesting from a physical perspective, our approach does yield generalizations of local Hamiltonian problems that have been independently studied by physicists.  

We introduce natural quantum generalizations of Vertex Cover and prove complexity and approximation results. Broadly, we suggest different methods for obtaining quantum generalizations of classical discrete optimization problems and consider quantum generalizations of Vertex Cover with respect to each.  
One of the generalizations is {\it Transverse Vertex Cover} (TVC) obtained by generalizing the objective function of Vertex Cover. 
The decision version of TVC lies in the complexity class Stoquastic MA (StoqMA), and we show that the optimization problem is StoqMA-hard. See \Cref{sec:stoq_VC} for the formal definition of TVC and proof of hardness.
We observe in \Cref{sec:PXP} that TVC generalizes the PXP model, which has been recently introduced in the study of Rydberg atom arrays.
TVC is also related to the {\it Transverse Ising model} (TIM) with certain restrictions on $ZZ$ and $Z$ terms through which Vertex Cover constraints can be enforced by making their strengths arbitrarily large. While TIM is known to be StoqMA-hard \cite{bravyi2017complexity} in general, our result shows the hardness for this more restricted case where some terms are necessarily ``infinitely strong''. 

We further generalize the constraints of Vertex Cover to obtain a problem we refer to as {\it Entangled Vertex Cover} (EVC). Somewhat counter-intuitively EVC is polynomial-time solvable on a classical computer in all cases which are not equivalent to TVC instances. Thus EVC is StoqMA-hard in the worst case. See \Cref{sec:entangled_VC} for more details on this. We also consider a stoquastic generalization of the Prize Collecting Vertex Cover problem, that is equivalent to the anti-ferromagnetic transverse Ising model.

We give a $(2+\sqrt{2})$-approximation algorithm for TVC using a generalization of the classical local ratio (LR) method (also known as the primal-dual schema).  A hallmark of LR is that while our approximation guarantee is with respect to a convex relaxation, we do not have to solve the relaxation. Our resulting algorithm is a simple deterministic linear-time classical approximation algorithm.  To the best of our knowledge, ours is the first application of LR to a quantum local Hamiltonian problem.  Classically, local-ratio or primal-dual methods have produced a range of approximation algorithms for diverse problems~\cite{bar2004local}.  We include a brief primer for LR as well as an explanation of how our quantum generalization arises from it.  We expect our quantum LR will inspire new quantum approximation algorithms, especially since currently known quantum approximation algorithms rely on limited techniques such as rounding semidefinite programming hierarchies~\cite{parekh2022optimal, parekh2021application}. Finally, we also demonstrate an LR-based 4.1938..-approximation algorithm for the unconstrained stoquastic Hamiltonian version of Vertex Cover which is equivalent to the anti-ferromagnetic transverse field Ising model after rearranging the local terms to make them positive semi-definite in a specific way.

\cnote{--new--}
Vertex Cover is a special case of TVC so we obtain a $2-\epsilon$ hardness of approximation under the Unique Games Conjecture for TVC, hence TVC is NP-hard to approximate to within a factor $2-\epsilon$ under the unique games conjecture~\cite{khot2008vertex}. Quantum analogs of the PCP theorem for StoqMA \cite{aharonov2019stoquastic} and QMA remain elusive, and it is possible that the study of hardness for approximating StoqMA problems will shed light on the hardness of approximating QMA-hard problems. We hope our work inspires broader efforts in porting techniques from classical approximation algorithms to quantum settings as well.

\subsection{Related work}

Generalizing classical constraint satisfaction problems to quantum setting has been previously done with regard to various classical problems. \cite{kitaev2002classical} showed that the local Hamiltonian problem is complete for the class QMA with later works improving upon the locality of the interaction terms in the Hamiltonian. In \cite{bravyi2011efficient}, Bravyi showed that the quantum 2-SAT problem, which is a generalization of the classical 2-SAT problem, is in P giving an efficient algorithm. We make use of some of the results from \cite{bravyi2011efficient} in our work in section \ref{sec:entangled_VC} to show that two orthogonality constraints with entangling states between pairs of qubits with a common qubit can be used to generate new constraints. 
The result for EVC we present is implicit in prior works ~\cite{bravyi2014bounds, laumann2009phase, aldi2021efficiently, de2010ground, ji2011complete} since it has already been observed that the feasible region of the problem we define is only polynomially large~\cite{de2010ground, ji2011complete} and that local observables can be efficiently calculated inside the feasible region ~\cite{de2010ground}.  However, specializing these known ideas to our context provides some interesting observations so we give a proof.
\subsection{Notation}


The Pauli matrices are defined as:
\begin{equation*}
\label{eq:paulis}
 \mathbb{I}=\begin{bmatrix}
1 & 0 \\
0 & 1
\end{bmatrix},
\,\,\,\,\,\,
X=\begin{bmatrix}
0 & 1 \\
1 & 0
\end{bmatrix},
\,\,\,\,\,\,
Y=\begin{bmatrix}
0 & -i \\
i & 0
\end{bmatrix}, \,\text{and}
\,\,\,\,\,\,
Z=\begin{bmatrix}
1 & 0 \\
0 & -1
\end{bmatrix}.
\end{equation*}
\noindent Subscripts indicate quantum subsystems among $n$ qubits.  For instance, the notation $\sigma_i$ is used to denote a Pauli matrix $\sigma \in \{X,Y,Z\}$ acting on qubit $i$, i.e., $\sigma_i := \mathbb{I} \otimes \mathbb{I} \otimes \ldots \otimes \sigma \otimes \ldots \otimes \mathbb{I} \in \mathbb{C}^{2^n \times 2^n}$, where the $\sigma$ occurs at position $i$.  We will say a Pauli operator on $n$ qubits is $t$-local if it is the product of $n$ Pauli matrices, at most of which $t$ are not equal to $\mathbb{I}$.  We will denote the $2-$qubit unitary which ``swaps'' the state between qubits as 
\begin{equation}
    SWAP=\begin{bmatrix}
        1 & 0 & 0 & 0\\
        0 & 0 & 1 & 0\\
        0 & 1 & 0 & 0\\
        0 & 0 & 0 & 1
    \end{bmatrix}
\end{equation}
If $G=(V, E)$ is a graph then an edge $(i, j)\in E$ will simply be denoted $ij\in E$.  A path in the graph will be denoted as $(v_1, v_2, ..., v_p)$ where $v_iv_{i+1} \in E$ for all $i$. Rest of the notation we use in this paper is either standard notation or should be clear from the context.

\section{Overview}

\subsection{Quantum versions of classical problems}

Our notion of a constrained local Hamiltonian problem is:
\begin{definition}[Constrained local Hamiltonian]
\label{def:constrained-LH}
For a $k$-local Hamiltonian $H = \sum_{S \subseteq V:|S| \leq k} H_S$, let
\cnote{changed from $\tr[C_l\rho]=b_l\rightarrow \tr[C_l\rho]=0$ since $b_l$ can be pushed into $C_l$}
\begin{align}
\mu^*(H, \{C_l\}) = \min\ &\tr[H\rho]\\
s.t.\qquad & \tr[C_l\rho]=0\quad \forall l, \label{def:constrained-LH:constraint} \\
& \tr[\rho] = 1,\\
& \rho \succeq 0.
\end{align}
\end{definition}
The Hamiltonian $H$ above is the \emph{objective}, and the $C_l$ are Hermitian \emph{constraint} operators that are assumed to have efficient classical descriptions.  While the $C_l$ need not be local or a projector, we expect that they will be for natural problems. The number of constraints will also generally be polynomial in the number of qubits. If each term $H_S$ of the objective and the constraints are diagonal in the computational basis, then the problem is a constrained classical constraint satisfaction problem (CSP).  In this case, the objective terms and constraints may be expressed as tensor products of $\mathbb{I}$ and $Z$ operators. Many interesting combinatorial problems can be framed with a distinctive feature that the ground space of the objective Hamiltonian $H$ is easy to find in the full Hilbert space, and it is easy to find states that satisfy the constraints $C_l$; however, computational hardness comes from the tension between the constraints and objective. Vertex cover provides a simple example, since for this problem $H$ is a $1$-local Hamiltonian and $\{C_l\}$ are frustration-free diagonal operators representing a monotonic feasible space (i.e., if a computational basis $\ket{x}$ is feasible, then any $\ket{x'}$ that flips a 0 in $x$ to 1 is also feasible).   

We describe three systematic approaches for obtaining (constrained) local Hamiltonian problems from (constrained) constraint satisfaction problems (CSPs). Start with a constrained classical CSP as described above. By lifting the restriction that each local objective term and the constraints be diagonal while retaining other problem-specific properties, we obtain a quantum local Hamiltonian problem from the CSP.  An example is the Quantum SAT problem~\cite{bravyi2011efficient}.  Classical SAT is captured by $k$-local terms that are each diagonal projectors of rank $2^k-1$ when restricted to the space of the $k$ qubits on which the term acts.  By considering general \emph{non-diagonal} projectors with the same rank condition, one obtains Quantum SAT.  We derive a quantum analogue of Vertex Cover, called Entangled Vertex Cover in this way.

For our second strategy, we add a transverse field term (i.e., $\sum_i w_i X_i$) to the objective of the classical CSP. For example, the classical Ising model gives rise to the transverse Ising model in this way.  This is how we derive the Transverse Vertex Cover problem from the Vertex Cover problem.  This strategy can be used to come up with new StoqMA-complete problems starting from known NP-complete problems.

For the final approach, we can consider the way in which the classical constrained CSP acts in the $Z$ basis and lift this to a quantum problem that acts analogously in the $X$ and $Y$ bases.  For example, the classical Ising model has $Z_i Z_j$ terms while the quantum Heisenberg model has $X_i X_j + Y_i Y_j + Z_i Z_j$ terms.  This is also exactly how the Max Cut and Quantum Max Cut~\cite{gharibian2019almost} problems are related. One way to formalize this procedure is to start with a natural SDP relaxation of the constrained CSP, which may be viewed as a vector program where there is a unit vector $v_i$ for each variable and the objective and constraints are functions of inner products of the vectors. We restrict the $v_i$ to be rank $3$ and then let each of the three components of $v_i$ correspond to the Paulis $X_i$, $Y_i$, and $Z_i$.  Now, the inner product $v_i \cdot v_j$ correspond to $X_i X_j + Y_i Y_j + Z_i Z_j$.  A unit vector $v_0$ corresponds to $\mathbb{I}$ so that $v_0 \cdot v_i$ corresponds to $\alpha X_i + \beta Y_i + \gamma Z_i$, for some choice of $\alpha,\beta,\gamma$ with $\alpha^2 + \beta^2 + \gamma^2 = 1$.  The SDP relaxation of the constrained CSP then corresponds to a constrained local Hamiltonian problem. 

We might generally expect NP-complete classical problems to give rise to StoqMA- or QMA-complete quantum local Hamiltonian problems using the strategies above. However, this is not always the case, as we see with Entangled Vertex Cover. Entangled, ``quantum'', constraints can be restrictive enough to force the feasible region to be only polynomially large, making a classically intractable problem easier in the quantum case.  For this reason, it is also interesting to study quantum generalizations where constraints are relaxed to soft constraints by adding them into the Hamiltonian with a penalty, as is done classically through ``prize collecting'' variants of problems. \cnote{--new--}We define the Prize Collecting variation of \Cref{def:constrained-LH} as follows.

\begin{definition}[Prize Collecting Constrained local Hamiltonian]
    \label{def:PC_constrained-LH}
    For a $k$-local Hamiltonian $H = \sum_{S \subseteq V:|S| \leq k} H_S$, with constraints $C_l$ and corresponding penalties $w_l \geq 0$, let
    \begin{align}
    \mu^*(H,\{C_l, w_l\}) = \min\ &\tr[H \rho] +\sum_l  w_l \tr[C_l \rho]\\
    s.t.\qquad & \tr[\rho] = 1,\\
    & \rho \succeq 0.
    \end{align}
\end{definition}

At a glance, this definition may seem no different than the local Hamiltonian problem but this perspective can help in coming up with new QMA-hard problems. For example, consider the ferromagnetic Heisenberg Hamiltonian whose ground state is trivial along with 1-local Hamiltonians that are easy to optimize on their own but when put together can create QMA-hardness \cite{schuch2009computational}.  In the classical case, the penalties $w_l$ can typically be set to polynomially large values to enforce constraints $C_l$; however, in the quantum case this is no longer true.

\subsection{Complexity of Vertex Cover generalizations}
Our first set of results concerns the complexity of a ``transverse'' generalization of the weighted Vertex Cover problem.  Given a graph $G=(V, E)$ and a set of weights $\{c_i: c_i\geq 0\,\, \forall \,\,i\in V\}$, weighted Vertex Cover is the problem of choosing a minimum weight set $S\subseteq V$ such that each edge $ij\in E$ has $i\in S$ or $j\in S$.  Vertex cover can be formulated as an optimization problem over classical (diagonal) density matrices.  For $x\in \mathbb{F}_2^{|V|}$ a computational basis state on $|V|$ qubits $\ket{x}$ corresponds to a subset of the vertices according to $\ket{x}\leftrightarrow S=\{i\in V: \,\,x_i=1\}$.  Demanding $ij\in E$ to be covered is the same as demanding that $ x_i$ or $x_j$ is equal to $1$, i.e., anything except $(x_i,x_j) = (0,0)$, and the problem is to minimize the value of a diagonal $1$-local observable over density matrices satisfying the constraints.
We define and initiate the study of several natural generalizations of this problem to non-diagonal objectives/constraints.

\paragraph{Tranverse Vertex Cover} We obtain the first quantum generalization by allowing non-diagonal $1$-local projectors $\phi_i$ for the objective such that Tr$[\phi_i Z_i] \leq 0$ while having the constraints unchanged.
We will refer to this as the {\it Transverse Vertex Cover} since it corresponds to adding Pauli $X$ terms just as in the Transverse Ising model.
\begin{restatable}[Transverse Vertex Cover]{definition}{tvc}
    \label{def:TVC}
    Given a vertex covering constraint graph $G(V,E)$ and a stoquastic 1-local Hamiltonian $H = \sum_{i \in V} c_i \phi_i$ such that $c_i \geq 0$, $\phi_i$ are 1-local projectors acting on $i^{th}$ qubit such that $\mathrm{Tr}[Z_i \phi_i] \leq 0 \ \forall i \in V$, Transverse Vertex Cover optimization problem is defined as  
    \begin{align}\label{eq:stoq_vc}
        &\min \mathrm{Tr}[H \rho]\\
        \label{eq:tvc_const}s.t.\,\,\,\,\,\,\,\,\,\,\, &\mathrm{Tr}[\ket{00}\bra{00}_{ij}\rho]=0 \,\,\,\,\, \forall ij\in E,\\
         &\mathrm{Tr}(\rho)=1,\\
         &\rho \semigeq 0.
    \end{align}
\end{restatable}

We use perturbative gadgets to show that estimating the output of the Transverse Vertex Cover problem is StoqMA-hard by a reduction from estimating the ground state energy of the Transverse Ising model to within a given additive precision that is inverse polynomial in the system size, which is known to be StoqMA-hard \cite{bravyi2017complexity}.

\begin{restatable}{theorem}{thmtvchardness}
    \label{thm:TVC_hardness}%
    Transverse Vertex Cover problem is StoqMA-hard.
\end{restatable}
In contrast to other reductions using perturbative gadgets \cite{cubitt2016complexity}, an important detail in our context is to ensure that the reduction is done strictly inside the subspace defined by the constraints.  We prove that the \cref{thm:TVC_hardness} is true even when we restrict the vertex cover constraint graph to have a maximum degree of 3. We also define a ``Prize Collecting'' version of Transverse Vertex Cover where edges can be uncovered while incurring additional cost in the objective.  This is easily seen to be equivalent to the anti-ferromagnetic transverse Ising model, hence is also StoqMA-hard.

\paragraph{Entangled Vertex Cover} The next generalization is to allow constraints that are not diagonal in the computational basis. Since the constraints for this problem involve projecting onto generic (possibly entangled) $2$-local projectors, we will refer to this as {\it Entangled Vertex Cover} (EVC).  The constraints for EVC generalize classical Vertex Cover, and instances of EVC with diagonal objective and constraints include Vertex Cover.  More precisely, EVC is defined by a 2-qubit projector $C$ that is (i) SWAP invariant and (ii) rank $1$. The constraint $\tr[C_{ij} \rho] = 0$, where $C_{ij}$ is $C$ acting on qubits $i$ and $j$, is imposed on each edge $ij$ of an input graph $G$.  The SWAP invariance of the constraints is tantamount to assuming that constraints treat endpoints of edges symmetrically.  While this is a generalization of Transverse Vertex Cover and hence also StoqMA-hard, known results~\cite{ji2011complete, de2010ground, bravyi2011efficient} imply that instances of the problem with ``entangled'' constraints are polynomial-time solvable in P.  Essentially an instances of this more general problem that does not differ from a Transverse Vertex Cover instance, up to local unitaries, corresponds to computing an extremal eigenvalue of a polynomially large matrix.  For completeness we give a proof of this fact by using the transfer matrix approach of Bravyi~\cite{bravyi2011efficient}.  We employ this method to derive new constraints from entangled constraints given in the problem definition.  We find that for a connected constraint graph the entangled constraints are strong enough to force the feasible region to depend \textbf{only} on the form of the constraint and the bipartiteness of the constraint graph. \cnote{--new--} \onote{The singlet example was not SWAP invariant, so I changed it.} An example of such a constraint would be a projector onto a SWAP-invariant Bell state between pairs of qubits $ij$ on a graph, e.g. $\tr[(I+X_iX_j+Y_iY_j-Z_iZ_j)\rho] = 0$. 

\onote{We need to reconcile this with the above, since EVC constraints are SWAP symmetric; I used ``a variant of this problem'' to indicate this below.}
While the `Entangled' constraints make the problem easy, when these constraints are relaxed to soft constraints following the \Cref{def:PC_constrained-LH}, a variant of this problem can be made QMA-hard.
\begin{theorem*}[\cite{schuch2009computational}]
Let $H$ be the ferromagnetic Heisenberg Hamiltonian with 1-local terms as follows
\begin{align}
    H = \sum_{ij} w_{ij} (I-X_iX_j-Y_iY_j-Z_iZ_j) + 1\text{-local terms}
\end{align}
where $w_{ij} \geq 0 \,\forall ij$. Estimating the ground state energy of $H$ to within a given inverse polynomial in the number of qubits additive precision is QMA-hard.
\end{theorem*}

\subsection{Quantum local ratio method}
The local ratio method is a general framework for approximating discrete optimization problems.  For a given problem, the local ratio method can be used to design an algorithm for generating an approximate solution as well as a technique for analyzing the corresponding approximation factor. 
 Given a problem instance the algorithm proceeds in rounds.  In each round, the algorithm selects a constraint and changes the instance to a simpler instance by modifying the objective function based on the constraint.  The algorithm terminates when the instance has a trivial solution and outputs such a solution. The quality of approximation depends only on the worst-case constraint, which is local, and this is where the method draws its name.  While the method is based on a mathematical programming formulation or relaxation of the problem, specifying the objective and constraints, this is never actually solved directly.  This makes for a powerful combination enabling simple and fast algorithms based on convex relaxations that do not need to be solved.  An equivalent perspective, called the primal-dual schema, shows that the local ratio method produces feasible (but not optimal) primal and dual solutions to a convex relaxation and bounds their gap to bound the approximation guarantee. 
 

\paragraph{Local ratio for Vertex Cover} We illustrate the local ratio method for the weighted Vertex Cover problem~\cite{bar2004local} by starting with a integer programming formulation for Vertex Cover:  
\begin{align}\label{eq:VC-integer-program-overview}
    &\min \sum_{i\in V} c_i x_i\\
     \label{cons:VC-overview}
     s.t.\,\,\,\,\,\,\,\,\,\,\, & x_i + x_j \geq 1 \quad \forall\, ij \in E\\  
    & x_i \in \{0,1\} \quad \forall\, i \in V.
    \label{cons:VC-integrality-overview}
\end{align}

Each variable $x_i$ indicates whether vertex $i$ is selected in a vertex cover.  The cost of the vertex cover is then linear in these variables, and the constraints ensure that each edge has at least one endpoint selected.  This NP-complete integer program precisely captures Vertex Cover, and it can be relaxed to a linear program, solvable in polynomial-time, by replacing Constraint~\eqref{cons:VC-integrality-overview} with $x_i \in [0,1]$ for all $i \in V$.

The local ratio method for Vertex Cover produces a feasible solution $\hat{x} = (\hat{x}_1,\ldots,\hat{x}_n)$ to the integer program above.  It first selects all vertices with zero cost, which does not affect the approximation ratio:
\begin{enumerate}
\item Set $\hat{x}_i = 1$ for all $i$ with $c_i = 0$.
\item Repeat until $\hat{x}$ is a vertex cover:
\begin{enumerate}
    \item\label{alg:local-ratio-step1} 
    Select an edge $ij \in E$ uncovered by $\hat{x}$.  Let $\epsilon_{ij} = \min\{c_i, c_j\}$, and modify the objective $c(x)$ to be $c(x) - \epsilon_{ij}(x_i + x_j)$.
    \item\label{alg:local-ratio-step2} 
    Set $\hat{x}_i = 1$ for all vertices $i$ such that $c_i = 0$.
\end{enumerate}
\end{enumerate}
The algorithm is guaranteed to cover at least one unconvered edge in each step, since Step~\eqref{alg:local-ratio-step1} ensures that for an uncovered $ij \in E$, either $c_i$ or $c_j$ becomes 0, and Step~\eqref{alg:local-ratio-step2} selects such vertices.  Thus the algorithm terminates in at most $|E|$ steps.

Let $x^*$ be an optimal solution. The algorithm gives the decomposition $c(x) = r(x) + \sum_{ij \in E} \epsilon_{ij}(x_i + x_j)$, where we take $\epsilon_{ij} = 0$ for edges not selected by the algorithm; the linear function $r(x)$ is the modified objective function upon termination of the algorithm, and we have $r(\hat{x}) = 0$.  Now consider the approximation ratio:
\begin{equation}
\alpha := \frac{c(\hat{x})}{c(x^*)} = \frac{\sum_{ij \in E} \epsilon_{ij}(\hat{x}_i + \hat{x}_j)}{r(x^*) + \sum_{ij \in E} \epsilon_{ij}(x^*_i + x^*_j)} \leq \frac{\sum_{ij \in E} \epsilon_{ij}(\hat{x}_i + \hat{x}_j)}{\sum_{ij \in E} \epsilon_{ij}(x^*_i + x^*_j)}.
\end{equation}
For each edge $ij$, the ``local ratio'' $(\hat{x}_i + \hat{x}_j)/(x^*_i + x^*_j)$ is at most 2, since $x^*_i + x^*_j \geq 1$ by Constraint~\eqref{cons:VC-overview} and $\hat{x}_i + \hat{x}_j \leq 2$.  Since the local ratio is at most 2 and $\epsilon_{ij} \geq 0$ for all $ij$, we get $\alpha \leq 2$.  

In retrospect, by the analysis of the local ratio above, we see that the modification to the objective function in Step~\eqref{alg:local-ratio-step1} comes directly from Constraint~\eqref{cons:VC-overview}.  More generally, the approximation guarantee of the local ratio method is determined by the worst-case local ratio between a feasible solution returned by the algorithm and an optimal solution with a respect to a constraint. 

Another byproduct of the analysis is that we could have taken $x^*$ to be the optimal solution to the linear program relaxation, since we only use that $x^*$ is nonnegative and satisfies Constraint~\eqref{cons:VC-overview}.  Therefore the above also shows that the worst-case gap between the linear program relaxation value and a solution produced by the algorithm is at most 2.

\paragraph{A quantum local ratio method} In this work we extend the approach above for constrained local Hamiltonian problems.  We are interested in minimizing a $k$-local Hamiltonian $H$ satisfying $H\succeq 0$ and subject to $2$-local constraints: $\tr[C_{ij}\rho]= 0\,\,\forall \,\,ij\in E$.  The basic idea is the same as for the classical Vertex Cover problem. In each round we subtract local terms $\epsilon_{ij} H_{ij}$ from $H$, ensuring $H\succeq 0$ for the resulting $H$. We maintain a product state solution and terminate when it satisfies all the constraints.  However, in our case $H_{ij}$ is not the same as $C_{ij}$, and bounding the local ratio is more involved.  In our case the local ratio for $ij$ is
\begin{equation}
\label{eq:quantum-local-ratio}
\frac{\max_{\rho : \tr[C_{ij} \rho]=0} \tr[H_{ij} \rho]}{\min_{\rho : \tr[C_{ij} \rho]=0} \tr[H_{ij} \rho]}.   
\end{equation}
While $\lambda_{\max}(H_{ij})/\lambda_{\min}(H_{ij})$ is an immediate upper bound, better approximation guarantees are obtained by directly considering \Cref{eq:quantum-local-ratio}.  In our case, we further improve this bound by observing the the maximum in the numerator can be restricted to the set of states $\rho$ that are output by our algorithm (product states in our case). 
Using this quantum local ratio method, we get a $(2+ \sqrt{2})$-approximation for TVC and a $4.194$-approximation for the Transverse Prize Collecting Vertex Cover.
\begin{restatable}{theorem}{LocalratioAlgForTVC}
    \Cref{alg:TVC} is a $(2+\sqrt{2})$-approximation algorithm for the Transverse Vertex Cover problem.
\end{restatable}

\begin{restatable}{theorem}{LocalratioAlgForTPCVC}
    \Cref{alg:SPCVC} is a 4.194-approximation algorithm for the Transverse Prize Collecting Vertex Cover problem.
\end{restatable}

\Cref{sec:QLR} provides a detailed account of our quantum local ratio method.

\onote{I was also thinking to keep things simple here and just describe the LR method for classical VC.  I think the description of the quantum version in Section 6 is reasonably short and clear, and we can just appeal to it as necessary, or maybe even move portions of it here, like the definitions. I'll look into this when I do my pass.}

\section{Discussion}
\subsection{PXP model as a special case of TVC}
\label{sec:PXP}

PXP model is a Hamiltonian commonly defined on a lattice in the context of Rydberg atoms. On a 1d chain, the Hamiltonian is
\begin{align}
    H_{\text{PXP, 1d}} = -\sum_i P_{i-1} X_i P_{i+1}
\end{align}
where $P_i = \ketbra{1}{1}_i$ with or without a boundary. The Hamiltonian $H_{\text{PXP}}$ is block diagonal in the computational basis where different blocks are identified by whether the reduced state of qubits $i,i+1$ are in the span of $\left\{\ket{01}, \ket{10}, \ket{11}\right\}_{i,i+1}$ or not. In a given computational basis state, the Hamiltonian $H_{\text{PXP}}$ acting on $\ket{00}_{i,i+1}$ state will leave it unaltered, effectively creating a boundary at $(i,i+1)$. Because of this, the subspace of primary interest for the PXP model is where all the neighboring qubit states are in the span of $\left\{\ket{01}, \ket{10}, \ket{11}\right\}$. This is exactly the Vertex Cover constraint in the computational basis. Unlike the transverse Ising model on a 1d chain which is integrable, the PXP model on a 1d chain has been proven to be non-integrable \cite{park2024proof}, and is known to host quantum scars. This phenomenon is of interest in the study of weak ergodicity breaking in non-integrable systems \cite{bernien2017probing, turner2018weak}. With the convention that a Rydberg state of a two-level atom is identified with $\ket{0}$, vertex cover constraints amount to not having two neighboring atoms to be both in Rydberg state $\ket{00}$. Then PXP Hamiltonian on a general graph $G(V,E)$ can be defined as
\begin{align}
    H_{\text{PXP, } G} = -\sum_i w_i \Pi_i X_i
\end{align}
where $w_i \geq 0\, \forall i \in V$ and $\Pi_i = \bigotimes_{\{j:\, ij \in E\}} \ketbra{1}{1}_j$ are the operators projecting all the neighbours of $i^{\text{th}}$ qubit to $\ket{1}$. PXP model can be seen as a special case TVC where the 1-local objective Hamiltonian, up to shifts of identity, is completely off-diagonal with only Pauli-$X$ terms. We conjecture that estimating the ground state energy of the PXP Hamiltonian in the subspace where all the vertex cover constraints are satisfied is StoqMA-hard.


\onote{Approximation algorithms and hardness are underdeveloped for StoqMA.  Is there is an analogue of PCP and Unique games hardness?  Cite Robbie's approximation (though EPR is not known to be StoqMA complete)}

\subsection{Open questions}

\begin{enumerate}

\item We prove that TVC is StoqMA-hard where the NP-hardness is explicit in the diagonal terms of the objective function since the classical Vertex Cover is a special case of this. Therefore, we already know that the problem is at least NP-hard and our result improves the lower bound from NP-hard to StoqMA-hard. What is the complexity of TVC in the special case where the objective function, up to shifts of identity, is completely off-diagonal with only Pauli-$X$ terms? This corresponds to PXP model without any additional 1-local Pauli-$Z$ terms and does not contain classical Vertex Cover as a special case.

\item Given that our approach to generalize Vertex Cover by having non-diagonal rank-1 projectors as constraints fails to produce a QMA-hard problem, what is a natural quantum generalization of Vertex Cover problem that is QMA-hard?

\item One can formulate relaxations for TVC based on the quantum Lasserre hierarchy. How do approximation algorithms based on rounding with respect to such hierarchies compare to the local ratio method? For classical Vertex Cover, it is known that the SDP-based rounding algorithms cannot do better than algorithms based on the local ratio method.

\item How can we further broaden the scope and efficacy of the quantum local ratio method, notably to produce entangled states? Can such approaches be used to obtain provably optimal approximation algorithms for the Transverse Vertex Cover problem, under the unique games conjecture?


\end{enumerate}

\subsection{A roadmap for the sequel} \Cref{sec:stoq_VC} introduces and proves StoqMA-hardness of the Transverse Vertex Cover problem and related problems.  Our quantum local ratio method for these problems is developed and analyzed in \Cref{sec:QLR}. Finally, the Entangled Vertex Cover is introduced and discussed in \Cref{sec:entangled_VC}, with a focus on the classical polynomial-time solvability of ``entangled'' instances. 
\section{Stoquastic generalization of Vertex Cover and Prize Collecting Vertex Cover}
\label{sec:stoq_VC}
\onote{We need to explain the classical VC and PVC problems.}
\onote{We should settle on conventions for problem names: e.g., Stoquastic Vertex Cover vs stoquastic cover vertex.  I prefer the former and am also fine using abbreviations like SVC}

\subsection{Vertex Cover and Prize Collecting Vertex Cover}
\label{sec:classical-VC}
The classical Vertex Cover problem is a well-known discrete optimization problem where given a graph $G(V,E)$ with positive weighted vertices $\{c_i \geq 0\}_{i \in V}$, the task is to pick a subset of vertices $S \subseteq V$ such that for every edge $ij \in E$ either $i \in S$ or $j \in S$ (or both) while minimizing the weighted sum of vertices picked to cover the edges. Vertex Cover is a special case of the Set Cover problem~\cite{vazirani2001approximation} and may be expressed as an integer program:
\begin{align}\label{eq:VC-integer-program}
    &\min \sum_{i\in V} c_i x_i\\
     \label{cons:VC}
     s.t.\,\,\,\,\,\,\,\,\,\,\, & x_i + x_j \geq 1 \quad \forall\, ij \in E\\  
    & x_i \in \{0,1\} \quad \forall\, i \in V.
\end{align}
The boolean variable $x_i$ indicates whether vertex $i$ is included in a vertex cover $S$.  The objective seeks to minimize the total cost of the vertices in $S$, and constraint \eqref{cons:VC} ensures that $S$ contains at least one endpoint of each edge.

The above integer program for Vertex Cover may be cast as a constrained local Hamiltonian problem: 
\begin{align}\label{eq:classic_vc}
&\min_{\rho} \sum_{i\in V} c_i \tr[\ketbra{1}{1}_i \rho]\\
\label{cons:Ham-VC} s.t.\,\,\,\,\,\,\,\,\,\,\, 
&\tr[\ketbra{00}{00}_{ij} \rho]=0 \,\,\,\,\, \forall ij\in E.\\
&\mathrm{Tr}(\rho)=1,\\ 
&\rho \semigeq 0.
\end{align}
To see that this captures Vertex Cover, first observe that since the objective operator $\sum_i c_i \ketbra{1}{1}_i$ and constraint operators $\ketbra{00}{00}_{ij}$ are diagonal in the computational basis, we may assume the same of an optimal $\rho$.  Thus $\rho$ is a mixture of basis states, $\sum_k \alpha_k \ket{x_{k,1}\ldots x_{k,n}}\bra{x_{k,1}\ldots x_{k,n}}$ with the $\alpha_k \geq 0$.  Constraint \eqref{cons:Ham-VC} ensures that constraint \eqref{cons:VC} holds for $x_{k,1},\ldots,x_{k,n}$ for each $k$, since for $ij \in E$,
\begin{equation}
\tr[\ket{00}\bra{00}_{ij}\rho] = 0 
\Rightarrow \bra{x_{k,1}\ldots x_{k,n}} \ket{00}\bra{00}_{ij} \ket{x_{k,1}\ldots x_{k,n}} = 0,\forall k
\Rightarrow x_{k,i} = 1 \text{ or } x_{k,j}=1,\forall k.  
\end{equation}
Likewise, the objective of the the constrained 2-local Hamiltonian problem \eqref{eq:classic_vc}, corresponds to the expectation of the objective of the integer program \eqref{eq:VC-integer-program}:
\begin{equation}
\sum_{i\in V} c_i \mathrm{Tr}[\ketbra{1}{1}_i \rho] = \sum_k \alpha_k \sum_{i\in V} c_i x_{k,i}.
\end{equation}
Thus we may assume that an optimal solution to \eqref{eq:classic_vc} is a pure basis state achieving minimum cost, giving us a direct correspondence between the two optimization problems above.  

The decision version of Vertex Cover was one of Karp's original 21 NP-complete problems and is known to be hard to approximate with an approximation ratio $2-\varepsilon$, for any constant $\varepsilon > 0$, under the Unique Games Conjecture~\cite{khot2008vertex}.  An unconstrained version of vertex cover known as Prize Collecting Vertex Cover is obtained by dropping the constraint that each edge in the graph is covered.  Instead, a penalty is imposed for not covering an edge, and the task is the minimize the total cost associated with picking the vertices plus the penalty for edges that are not covered by the picked vertices. Given a graph $G(V,E)$ with cost $c_i \geq 0$ for picking vertex $i \in V$ and $c_{ij} \geq 0$ for not picking edge $ij \in E$, the corresponding classical Hamiltonian can be written as:
\begin{align}
    \quad \sum_{i \in V} c_i \ket{1}\bra{1}_i + \sum_{ij\in E} c_{ij} \ket{00}\bra{00}_{ij} = \sum_{i \in V} c_i \frac{(I-Z_i)}{2} + \sum_{ij \in E} c_{ij} \frac{(I+Z_i)\otimes(I+Z_j)}{4}.
\end{align}

Finding the ground state energy of this diagonal Hamiltonian is also NP-hard, and there are various algorithms based on linear programming, the local ratio (or primal-dual) method, and semi-definite programming \cite{Has06, PANDEY2018121} that achieve an approximation ratio of $2$ for both Vertex Cover and Prize Collecting Vertex Cover. An easy way to see that hardness results above also hold is by reducing Vertex Cover to the Prize Collecting version by setting each penalty $c_{ij} = \max\{c_i, c_j\}$, in which case it is never beneficial to pay any penalties.

\subsection{Transverse Vertex Cover}\label{sec:tvc_complexity}

Towards a quantum generalization of the classical vertex cover problem, we consider the problem where the edge covering constraints remain the same as in the classical vertex cover but the objective Hamiltonian is a sum of single qubit rank-1 projectors. 
For any 1-local objective Hamiltonian, we can apply single-qubit unitary rotations around the $Z$-axis on the Bloch sphere to rotate each $\phi_i$ so that it does not have any $Y$ component. Since these rotations are around the $Z$-axis of each qubit, the $Z$-basis constraints are not affected. Further, we can conjugate the necessary qubits with $Z$ to change $X \rightarrow -X$ to make the objective Hamiltonian stoquastic. Therefore it is enough to consider only stoquastic 1-local Hamiltonian as the objective, i.e.,  Tr$(X_i \phi_i) \leq 0$ and Tr$(Y_i \phi_i) = 0$. We formally define the Transverse Vertex Cover problem as follows.

\tvc*

Similar to the case of ground state energy estimation in the local Hamiltonian problem, we are interested in estimating the optimal value of TVC to within a given inverse polynomial precision in the number of qubits. The decision version of TVC optimization problem is in the complexity class StoqMA. Bravyi and Hastings \cite{bravyi2017complexity} showed that estimating the ground state energy of TIM is hard for the class StoqMA. We use perturbative gadgets to show that the ground state energy of TIM can be simulated using TVC thus proving that TVC is StoqMA-hard.
The perturbative gadget we use is based on the Bloch expansion similar to that of \cite{jordan2008perturbative}. We review the perturbative Bloch expansion that is necessary for our proof in Appendix \ref{appx:bloch_expansion}. The general idea behind these kinds of reductions is to efficiently embed the spectrum of a given Hamiltonian (TIM in our case) as the low-energy spectrum of another Hamiltonian in a larger Hilbert space (TVC in our case). We can achieve this by starting from an unperturbed Hamiltonian $H_0$ whose ground space is degenerate and has the dimension equal to the dimension of the Hilbert space whose spectrum we want to embed. Then we add a perturbation term $V$ that splits the degeneracy and achieves the spectrum that we want. The term \emph{gadget} in perturbative gadget involves designing the unperturbed Hamiltonian $H_0$ and the perturbation $V$ that achieves the goal of efficiently embedding the spectrum with an arbitrary additive precision.

\thmtvchardness*

\begin{proof}
    Let $G(V,E)$ be a maximum degree-3 interaction graph of the TIM with the Hamiltonian $H_{\text{TIM}} = \sum_{ij \in E}w_{ij}Z_iZ_j + \sum_{i \in V} h_i X_i$, where $h_i \leq 0 \ \forall i \in V$ whose ground state energy we want to estimate within a precision $\epsilon = 1/\text{poly}(|V|)$. 
    
    We start by encoding the logical qubits from the TIM into the lowest energy eigenspace of a 1-local diagonal Hamiltonian restricted to the vertex covering Hilbert subspace. For every qubit $i$ in the TIM Hamiltonian, we will have two qubits $i$ and $\tilde{i}$ in the TVC with an objective Hamiltonian term 
    \begin{align} \label{eq:TVC_H0_i}
        (H_0)_i = -\frac{7}{2}(Z_i + Z_{\tilde{i}})
    \end{align}
    and an edge between $i$ and $\tilde{i}$ in the constraint graph. Since an edge restricts the Hilbert space to be the span of $\left\{\ket{01}, \ket{10}, \ket{11}\right\}$ for $i$ and $\tilde{i}$, the ground subspace in the restricted Hilbert space of coverings is a span of $\left\{\ket{01}, \ket{10}\right\}$, which we map to $\ket{0}$ and $\ket{1}$ states respectively of qubit $i$ in the TIM. This representation is analogous to the dual-rail representation that was used in many previous works including in \cite{bravyi2017complexity} to encoded qubits using hardcore bosons.
    
    For every pair of qubits $ij$ in the TIM with a $ZZ$ interaction, we create a gadget that adds 4 additional edge qubits with labels $(a,b)$ where $a \in \{i, \tilde{i}\}$ and $b \in \{j, \tilde{j}\}$. The qubit labeled with $(a,b)$ has edges connecting it with qubit $a$ and $b$ in the constraint graph. \Cref{fig:tvc_gadget} illustrates the covering constraint graph for 2 encoded qubits. We also add the following objective Hamiltonian term 
    \begin{align}\label{eq:TVC_H0_ij}
        (H_0)_{ij} = -\frac{1}{2}\left( Z_{(i,j)} + Z_{(i,\tilde{j})} + Z_{(\tilde{i},j)} + Z_{(\tilde{i},\tilde{j})} + 2I\right)
    \end{align}
    for the edge qubits. The total Hamiltonian $H_0$ is the sum of Hamiltonian terms for vertex qubits and edge qubits,
    \begin{align} \label{eq:unperturbed_tvc_H}
        H_0 = \sum_{i \in V} (H_0)_i + \sum_{ij \in E} (H_0)_{ij},
    \end{align}
    where $(H_0)_i$ and $(H_0)_{ij}$ are defined as in \cref{eq:TVC_H0_i} and \cref{eq:TVC_H0_ij} respectively. 

    \begin{figure}[!ht]
        \centering
        \captionsetup{justification=centering}
        \includegraphics[scale=1]{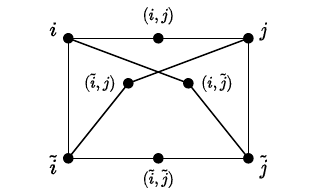}
        \caption{Covering constraint graph of the physical qubits in TVC that encode two qubits $i$ and $j$ from the TIM with a $ZZ$ interaction.}
        \label{fig:tvc_gadget}
    \end{figure}

    Given that we are only considering maximum degree-3 TIMs, the factor of $7$ in front of $(Z_i + Z_{\tilde{i}})$ is sufficient to ensure that the lowest energy eigenstates of $H_0$ in the restricted vertex covering Hilbert subspace have only one of the qubits out of the two qubits $i, \tilde{i}$ in $\ket{1}$ state but not both for every $i$. This also implies that out of the four edge qubits used to encode a $ZZ$ interaction from the TIM, the lowest energy eigenstates have only one of the edge qubits per gadget in $\ket{0}$ state while the other three are in $\ket{1}$ state as illustrated in  \cref{fig:tvc_gadget_basis_states} for two encoded qubits. 
    
    \begin{figure}[!ht]
        \centering
        \captionsetup{justification=centering}
        \includegraphics[scale=0.8]{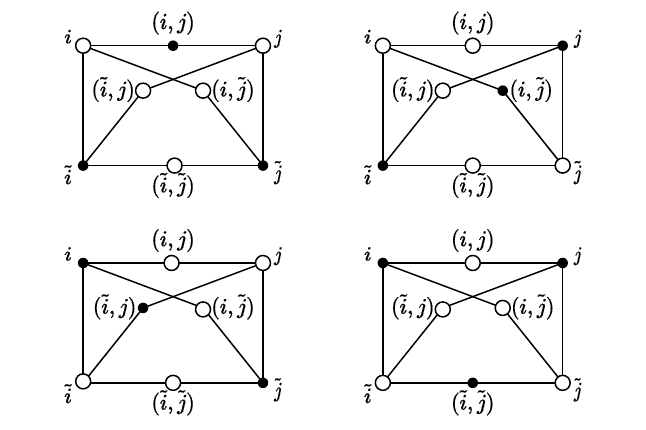}
        \caption{Four basis states of TVC gadget that encodes two qubits. The light-colored circles indicate that the physical qubit is in $\ket{1}$ state while the dark-colored circles indicate that the qubit is in $\ket{0}$ state.}
        \label{fig:tvc_gadget_basis_states}
    \end{figure}
    
    Following the above argument, as we encode a total of $n$ qubits, the dimension of the lowest energy eigenspace of $H_0$ in the restricted covering Hilbert subspace is $2^n$ with an eigenenergy zero, and each computational basis state in that zero energy eigenspace can be mapped back to computational basis states in the TIM based on which of $i$ or $\tilde{i}$ is in $\ket{1}$ state. 
    
    Let us consider the objective Hamiltonian $H_0$ we have so far to be the unperturbed Hamiltonian. To add a $ZZ$ interaction from the TIM, we add a 1-local $Z$ perturbative Hamiltonian term to the edge qubits as
    \begin{align}
        (V_{ZZ})_{ij} = \frac{w_{ij}}{2}(Z_{(i,j)} + Z_{(\tilde{i}, \tilde{j})} - Z_{(i, \tilde{j})} - Z_{(\tilde{i}, j)})
    \end{align}
    where $w_{ij}$ is the weight of the $Z_iZ_j$ interaction in TIM. At first-order perturbation, this gives rise to $ZZ$ interactions between the encoded logical qubits in TVC as
    \begin{align}
        \widetilde{H}_{\text{eff}}^{(1)} = w_{ij}\overline{Z}_i\overline{Z}_j
    \end{align}
    where $\overline{Z}_i$ is the Pauli-$Z$ operator on the encoded $i^{th}$ qubit in TVC. 
    
    To encode the transverse field from TIM, first, let us consider the case where the underlying interaction graph of TIM is 3-regular. Later, we will see how to deal with the case of graphs having degree-2 and degree-1 vertices. Starting from a 3-regular interaction graph of TIM and going through the above gadget reduction to encode the qubits, the encoded logical basis states are hamming distance 8 away from their neighbors in the physical qubit space as depicted in \cref{fig:tvc_flip_qubit}: distance 2 to flip the encoded qubit in $i$ and $\tilde{i}$, plus an additional distance 6 to flip the 2 edge qubits per edge to reach a neighboring encoded basis state. This is equivalent to flipping the $i^{th}$ qubit in TIM. 
    \begin{figure}[!ht]
        \centering
        \captionsetup{justification=centering}
        \includegraphics[scale=0.9]{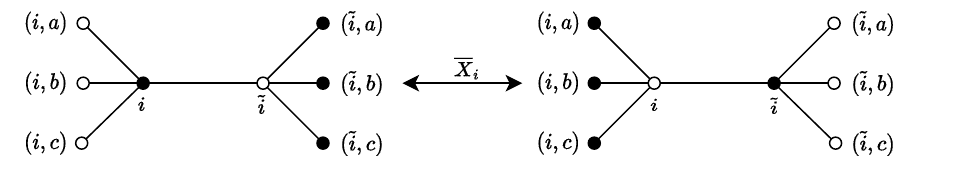}
        \caption{The above figure shows the action of flipping the encoded qubit $i$ in TVC between states $\ket{0}$ and $\ket{1}$. The light-colored circles indicate that the physical qubit is in $\ket{1}$ state while the dark-colored circles indicate that the qubit is in $\ket{0}$ state.}
        \label{fig:tvc_flip_qubit}
    \end{figure}
    Therefore, we can encode the transverse field for each qubit as an $8^{th}$ order perturbative interaction in the TVC. Consider the  perturbation terms as
    \begin{align}
        (V_{X})_i &= -\sqrt{4h_i}\ (X_i + X_{\tilde{i}}) \\
        (V_{X})_{ij} &= -(X_{(i,j)} + X_{(i,\tilde{j})} + X_{(\tilde{i},j)} + X_{(\tilde{i},\tilde{j})})
    \end{align}
    where $h_i$ is the strength of the transverse field for the $i^{th}$ qubit in TIM. The final form of the Hamiltonian after adding the perturbative terms is
    \begin{align}
        H = \Delta H_0 + V_{ZZ} + \Delta^{\frac{7}{8}} V_X, \nonumber
    \end{align}
    \begin{align}
        \text{where}\quad V_{ZZ} = \sum_{ij \in E} (V_{ZZ})_{ij} \quad \text{and} \quad V_X = \sum_{i \in V} (V_X)_i + \sum_{ij \in E} (V_X)_{ij}.
    \end{align}
    In the Bloch expansion, there exists a term at $8^{th}$ order of the form
    \begin{align} \label{eq:8th_order_term}
        P_0(V_{X}H_0^{-1})^7V_{X}P_0
    \end{align}
    which can encode the operators proportional to $\overline{X}_i$. Here $P_0$ is the projector on to the ground subspace of $H_0$ in the restricted Hilbert subspace of coverings.  Note that there are multiple combinatorial ways, specifically there are $69\times 3!\times 3!$ ways, to flip the encoded $i^{th}$ qubit in TVC when the degree of $i^{th}$ qubit in TIM is 3. In all these ways of flipping the basis (low energy) state, the excited states in between are also coverings since we are only considering the restricted covering Hilbert subspace. This is an important restriction as the final ground state will have non-zero overlap with the excited states and for the final ground state to satisfy the covering constraints, we need the excited also to satisfy the covering constraints. When we sum over all the ways we can flip the encoded $i^{th}$ qubit along with the weights that come from $H_0^{-1}$ in between $V$'s in expression (\ref{eq:8th_order_term}), we get an effective term that is equal to $h_i\overline{X}_i$.

    Suppose that there are some qubits in TIM with degree-2 and degree-1 $ZZ$ interaction. Under our gadget reduction, this would lead to encoded basis states that are Hamming distance 6 and 4 away corresponding to flipping those encoding qubits. To adjust to this, we can change the perturbation order at which the transverse field is encoded for these particular qubits by changing the strength of the perturbation. For example, let $i^{th}$ qubit has degree-2 in TIM. Instead of having $\Delta^{\frac{7}{8}}\sqrt{4h_i}(X_i + X_{\tilde{i}})$ as a perturbation term, we would have $\Delta^{\frac{3}{4}}\sqrt{5h_i}(X_i + X_{\tilde{i}})$ which would lead to transverse field on the encoded $i^{th}$ qubit arising at $6^{th}$ order instead of $8^{th}$. Similarly for qubits with degree-1, by changing the perturbation strength from $\Delta^{\frac{7}{8}}\sqrt{4h_i}$ to $\Delta^{\frac{5}{8}}\sqrt{6h_i}$, we can encode the transverse field as a $4^{th}$ perturbative interaction.
    
    Note that there will be non-trivial energy shifts of low energy Hilbert space into which we encoded our qubits at orders $ \leq 8^{th}$ of perturbation due to $V_X$. This shift is computable in polynomial time given the graph $G(V, E)$ and $H_{\text{TIM}}$. A similar notion of energy shift also arises when encoding a $k$-body Pauli interaction directly as a $k^{th}$ order perturbative interaction as in \cite{jordan2008perturbative}. With the energy shift corrected, by having a large enough $\Delta$ which is polynomial in the number of qubits, the lowest expectation energy of $H$ in the restricted covering Hilbert subspace can be made inverse polynomially close to the ground state energy of $H_{\text{TIM}}$. Since estimating the ground state energy of TIM to within a given inverse polynomial additive precision in StoqMA-hard, so is TVC.
\end{proof}
In the proof of \Cref{thm:TVC_hardness}, we show the hardness of TVC where the constraint graph has a maximum degree of 7. In \Cref{thm:TVC_hardness_deg_3}, we improve upon the degree and show that the TVC problem is StoqMA-hard even when the maximum degree of the constraint graph is 3, using a gadget that is similar to the one that has been previously used in \cite{alimonti2000some} to show the hardness classical Vertex Cover on maximum degree-3 graphs.

\begin{theorem}\label{thm:TVC_hardness_deg_3}
    Transverse Vertex Cover problem is StoqMA-hard even when the constraint graph has a maximum degree of 3.
\end{theorem}

\begin{proof}
    Consider an instance of TVC where the maximum degree of the constraint graph is $d > 3$, and let $v$ be a qubit whose degree is equal to $d$ in the constraint graph. We will construct a new instance of TVC where qubit $v$ is replaced with $3$ qubits $v_a, v_b, v_c$ and split the edges that are connected to qubit $v$ between $v_a$ and $v_c$ while adding two additional edges: one between $v_a$ and $v_b$, and another between $v_b$ and $v_c$. Figure \ref{fig:tvc_degree_reduction} gives a pictorial representation of this process.
    \begin{figure}[!ht]
        \centering
        \captionsetup{justification=centering}
        \includegraphics{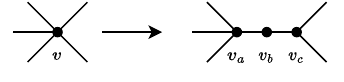}
        \caption{Gadget replacing a higher degree vertex with three vertices in a line each with a lower degree than the original vertex.}
        \label{fig:tvc_degree_reduction}
    \end{figure}

    To encode the qubit $v$, we add an unperturbed Hamiltonian term 
    \begin{align}
        H_0 = -\frac{\Delta}{2}\left(Z_{v_a}+2Z_{v_b}+Z_{v_c}\right)
    \end{align}
    whose ground states are $\ket{010}$ and $\ket{101}$ in the restricted covering Hilbert subspace which we map to $\ket{0}$ and $\ket{1}$ states of qubit $v$. We can encode $Z_v$ by adding a first-order perturbation term 
    $V_Z = \frac{1}{2}(Z_{v_a} + Z_{v_c})$, and we can encode $X_v$ by adding a third order perturbation term $V_X = \Delta^{2/3}(X_{v_a} + X_{v_b} + X_{v_c})$. Note that there will be an energy shift in the low energy eigenspace of $H_0$ due to second-order perturbative correction. Accounting for this shift, the ground state energy of this new instance of TVC can be made inverse polynomially close to the ground state energy of the given instance of TVC by having a large $\Delta$ that scales polynomially in the number of qubits. By repeatedly applying this procedure a constant number of times in parallel, we can reduce the maximum degree of TVC from 7 to 3. Therefore TVC is StoqMA-hard even when the maximum degree of the constraint graph is 3.
\end{proof}

\subsection{Transverse Prize Collecting Vertex Cover}
\begin{definition}
    \label{def:transverse_PCVC}
    We define the Transverse Prize Collecting Vertex Cover Hamiltonian from classical Prize Collecting Vertex Cover where single qubit projectors $\ket{1}\bra{1}$ are replaced by 1-local projectors as follows:
    \begin{align}
        H = \quad \sum_{i \in V} c_i \phi_i + \sum_{ij \in E} c_{ij} \frac{(I+Z_i)\otimes(I+Z_j)}{4} \label{eq:Ham_stoq_PCVC}
    \end{align}
where $\phi_i$ is 1-local projector on $i^{th}$ qubit with Tr$(Z_i \phi_i) \leq 0$, $c_{i} \geq 0\, \forall\, i \in V$ and $c_{ij} \geq 0\, \forall\, ij \in E$.
\end{definition} 
Similar to the case of TVC, we can assume that the 1-local projectors $\phi_i$ are stoquastic. About the complexity of computing the ground state energy, we can convert any anti-ferromagnetic TIM Hamiltonian into the above form, and therefore estimating the ground state energy is StoqMA-hard \cite{bravyi2017complexity}.
\section{A Quantum local ratio method}
\label{sec:QLR}

A natural question we consider is how well can we approximate Transverse Vertex Cover and Transverse Prize Collecting Vertex Cover problems, first with product states and in general with any quantum state. In this section, we generalize the well-known approximation algorithms based on the local ratio method to give product state approximations to the above problems.

We motivate the local ratio method in the quantum setting and illustrate it using classical Vertex Cover.  We start by specializing \Cref{def:constrained-LH} to the constrained 2-local Hamiltonian problem:

\begin{definition}[Constrained $2$-local Hamiltonian] For a $2$-local Hamiltonian $H \succeq 0$, let
\begin{align}\label{eq:cons-local-Ham}
\mu^*(H, \{C_{ij}\}) = \min\ &\tr[H\rho]\\
\label{ineq:LR-constraint_v1}
s.t.\qquad & \tr[C_{ij}\rho]=0\quad \forall ij \in E,\\
& \tr[\rho] = 1,\\
& \rho \succeq 0.
\end{align}
We use $\mu^*(H, \{C_{ij}\})$ to refer to both the problem above as well as its optimal value, and this distinction should be clear from the context.  
\end{definition}

For the problems we consider, the constraints will correspond to restriction to subspaces so that the $C_{ij}$ will be projectors.

\begin{definition}[$\alpha$-approximation] For $\alpha \in [1,\infty)$, the state $\rho$ is an \emph{$\alpha$-approximation} for the constrained local Hamiltonian problem \eqref{eq:cons-local-Ham} if
\begin{equation}
    \tr[H\rho] \leq \alpha \mu^*,
\end{equation}
where $\mu^* = \mu^*(H, \{C_{ij}\})$ is the optimal value. A polynomial-time algorithm producing a description of an $\alpha$-approximate state\footnote{A quantum approximation algorithm may instead opt to prepare the output state.} is called an \emph{$\alpha$-approximation algorithm}. The parameter $\alpha$ is called the \emph{approximation ratio} or \emph{approximation guarantee}.
\end{definition}

We seek to obtain an approximation algorithm with a good (i.e., relatively small) approximation ratio. We will accomplish this by decomposing $H$ into a linear combination, with nonnegative coefficients, of terms $H_{ij}$ acting on qubits $i$ and $j$.  These terms will be chosen so that constrained local Hamiltonian problem, $\mu^*(H_{ij}, \{C_{ij}\})$ becomes easy to solve when the objective $H$ is replaced with $H_{ij}$.  In fact $\mu^*(H_{ij}, \{C_{ij}\})$ will become an entirely \emph{local} problem, only depending on qubits $i$ and $j$.  We will use $\mu^*_{ij} = \mu^*(H_{ij}, \{C_{ij}\})$ for simplicity.  In particular, we seek to find coefficients $w_{ij} \geq 0$ so that
\begin{equation} \label{eq:QLR-decomp}
    H = R + \sum_{ij \in E} w_{ij} H_{ij},
\end{equation}
where $R \succeq 0$ is a ``remainder'' term.  The idea is then that any state that attains an $\alpha$-approximation on $R$ and each local $H_{ij}$ must be an $\alpha$-approximation for the global $H$.

\begin{theorem}[Local ratio] \label{thm:local_ratio}
Suppose there is a state $\rho$ such that    
\begin{gather}
\label{eq:thm-QLR1} \tr[R \rho] = 0,\text{ and }\\
\label{eq:thm-QLR2} \tr[H_{ij} \rho] \leq \alpha \mu^*_{ij},\ \forall ij \in E.
\end{gather}
Then $\rho$ is an $\alpha$-approximation for $H = R + \sum_{ij \in E} w_{ij} H_{ij}$, where $w_{ij} \geq 0$ for all $ij \in E$.
\end{theorem}
\begin{proof} Let $\rho^*$ be an optimal state for problem \eqref{eq:cons-local-Ham} so that $\tr[H \rho^*] = \mu^*$. We get
\begin{align}
\label{eq:QLR1} \tr[H\rho] &= 0 + \sum_{ij \in E} w_{ij} \tr[H_{ij} \rho]\\
\label{eq:QLR2} & \leq \alpha \sum_{ij \in E} w_{ij} \mu^*_{ij} \\
\label{eq:QLR3} & \leq \alpha \sum_{ij \in E} w_{ij} \tr[H_{ij} \rho^*]\\
\label{eq:QLR4} & = \alpha \tr[H \rho^*] = \alpha \mu^*. 
\end{align}
\Cref{eq:QLR1,eq:QLR2} follow from \cref{eq:thm-QLR1,eq:thm-QLR2}. The inequality \eqref{eq:QLR3} follows since the global optimal solution, $\rho^*$, is a feasible solution to the local problem $\mu^*(H_{ij}, \{C_{ij}\})$, so an optimal local solution must have objective at least as good as $\rho^*$. 
\end{proof}

\subsection{Local ratio algorithms for Transverse Vertex Cover problems}

How can we use \Cref{thm:local_ratio} to design approximation algorithms?  The key ingredient is selecting a local terms $H_{ij}$ on each edge $ij$ for which it is easy to obtain a good approximation (i.e., satisfying \cref{eq:thm-QLR2}).  Classically this is usually done by selecting an unweighted sum of the objective terms on $i$ and $j$, and we show that this strategy is also effective in the transverse case, though it takes more analysis to do so.  

Our quantum implementation of the local ratio method for Transverse Vertex Cover will construct a product state $\rho = \otimes_i \rho_i$. We first look for any qubit $i$ with $c_i = 0$, and in this case we may satisfy all constraints involving $i$ by setting $\rho_i = \ketbra{1}{1}_i$.  We say that all edges incident to $i$ are \emph{covered}.  We iteratively obtain a decomposition of $H$ into the form \eqref{eq:QLR-decomp} by picking an \emph{uncovered} edge $ij$ and subtracting $w_{ij} H_{ij}$ from $H$ for a weight $w_{ij} > 0$.  We pick $w_{ij}$ so that in $H' := H - w_{ij} H_{ij}$ either $c'_i = 0$ or $c'_j = 0$.  We then repeat this algorithm until all edges are covered, and $R$ is the portion of $H$ that remains.  All edges are covered at this point, and no constraints remain; we may simply set $\rho_i$ for all qubits $i$ on which $R$ acts so that the cost paid on $R$ is 0.

\begin{algorithm}[!h]
\setstretch{1.5}
\caption{Local ratio product state algorithm for Transverse Vertex Cover}
\label{alg:TVC}
\vspace{0.1in}
{\bf Given}: Graph $G(V,E)$ and a 1-local Hamiltonian $H = \sum_i c_i \phi_i$.\\
{\bf Assumptions}: $\phi_i$ is a rank-1 projector, $c_i \geq 0$, $\mathrm{Tr}[Z_i \phi_i] \leq 0\, \forall i \in V$.
\begin{algorithmic}[1]
\State For any $i \in V$ with $c_i =0$, set $\rho_i = \ketbra{1}{1}$
\State Choose an edge $ij \in E$ such that $\min\{c_i, c_j\} > 0$, and let $H_{ij} = \phi_i + \phi_j$
\State Update $H \rightarrow H - w_{ij}H_{ij}$ where $w_{ij} = \min\{c_i, c_j\}$

\State Repeat from step 1 until no more edges are left to choose in step 2.

\State For any remaining $i$ with $c_i > 0$, set $\rho_i = I-\phi_i$.
\end{algorithmic}
{\bf Output: } Product state $\rho = \otimes_i \rho_i$.
\vspace{0.075in}
\end{algorithm}


\LocalratioAlgForTVC*

\begin{proof}
\onote{We should also explain that we produce a feasible solution. This is described in the paragraph above the Algorithm, and we can be brief here.}

    To see that \Cref{alg:TVC} outputs a feasible solution for TVC, observe that at no point in the algorithm are we changing the constraint graph and in step 1 of each iteration of the algorithm we are setting at least one of $\rho_i$ or $\rho_j = \ketbra{1}{1}$ for the edge $ij$ picked in the previous iteration and the algorithm ends only after all the edge constraints are satisfied.

     For $ij\in E$, let $H_{ij} = \phi_i + \phi_j$ and let $\mathrm{Tr}[Z_i\phi_i] = \alpha$ and $\mathrm{Tr}[Z_j\phi_j] = 2k-\alpha$ where $-1 \leq \alpha \leq 0$ and $-1 \leq 2k-\alpha \leq 0$ without loss of generality. We may also assume $Tr[\phi_\ell Y_\ell]=0$ and $Tr[\phi_\ell X_\ell] \leq 0$ for $\ell\in \{i, j\}$ WLOG by the discussion in \Cref{sec:tvc_complexity}. By \Cref{thm:local_ratio}, it is enough to upper bound the ratio of the maximum expected value from our ansatz to the minimum objective of $H_{ij}$ subject to the constraints: $Tr[H_{ij} \rho]/\mu_{ij}^*\leq \alpha$.  Since the Vertex Cover constraint \cref{eq:tvc_const} forces $Tr[\ket{00}\bra{00}_{ij}\rho]=0$ we can compute $\mu_{ij}^*$ by optimizing $H_{ij}$ in the  subspace spanned by $\{\ket{01}, \ket{10}, \ket{11}\}$.  Denote $\widetilde{H_{ij}}$ as $H_{ij}$ restricted to this subspace.
    \begin{align}
        \widetilde{H_{ij}} =
        \begin{bmatrix}
            1+\alpha-k & 0 & -\frac{\sqrt{1-\alpha^2}}{2}\\
            0 & 1+k-\alpha & -\frac{\sqrt{1-(2k-\alpha)^2}}{2}\\
            -\frac{\sqrt{1-\alpha^2}}{2} & -\frac{\sqrt{1-(2k-\alpha)^2}}{2} & 1-k
        \end{bmatrix}.
    \end{align}
    The smallest eigenvalue is $\frac{1}{2}\left(-\sqrt{2 \alpha^2-4 \alpha k+k^2+2}-k+2\right)$ which is minimized with respect to $\alpha$ when $\alpha = 0$ or $2k$ and is equal to $\frac{1}{2}\left(-\sqrt{k^2+2}-k+2\right)$.

    To complete the analysis we need an upper bound on $Tr[H_{ij} \rho]$ where $\rho$ is the state obtained by the algorithm.  At the end of the algorithm, qubit-$i$ has two possible outcomes $I-\phi_i$ or $\ketbra{1}{1}$, and every pair of qubits on an edge $ij \in E$ have three possible outcomes $\left\{(I-\phi_i)\otimes \ketbra{1}{1}, \ketbra{1}{1}\otimes(I-\phi_j), \ketbra{1}{1}\otimes \ketbra{1}{1}\right\}$. Given that $\mathrm{Tr}[Z_i\phi_i] \leq 0 \, \forall i$, the maximum expectation value of $H_{ij}$ is with state $\ketbra{1}{1}\otimes \ketbra{1}{1}$ which is equal to 
    \begin{align}
        \mathrm{Tr}[H_{ij}\ketbra{1}{1}\otimes\ketbra{1}{1}] = 1 - \frac{(\mathrm{Tr}[Z_i\phi_i] + \mathrm{Tr}[Z_j\phi_j])}{2}=1-k.
    \end{align}
    For a fixed value of $k$, the approximation ratio is equal to
    \begin{align}
        \frac{2 - (\mathrm{Tr}[Z_i\phi_i] + \mathrm{Tr}[Z_j\phi_j])}{\left(-\sqrt{k^2+2}-k+2\right)} = \frac{2(1-k)}{\left(-\sqrt{k^2+2}-k+2\right)}
    \end{align}
    whose maximum value is equal to $2+\sqrt{2}$ which happens at $k=0$ in the range $-1 \leq k \leq 0$.
\end{proof}

\paragraph{Transverse Prize Collecting Vertex Cover} 
From \cref{def:transverse_PCVC}, the Transverse Prize Collecting Vertex Cover Hamiltonian is 
\begin{align}
    H = \quad \sum_{i \in V} c_i \phi_i + \sum_{ij \in E} c_{ij} \frac{(I+Z_i)\otimes(I+Z_j)}{4}
\end{align}
where $\phi_i$ is a 1-local projector with $\tr[Z_i\phi_i] \leq 0$, $c_i \geq 0\, \forall i \in V$ and $c_{ij} \geq 0\, \forall ij\in E$.
Our quantum local ratio implementation in this case is similar to \Cref{alg:TVC}. However, we have an unconstrained local Hamiltonian, so $H_{ij}$ contains both 1- and 2-local terms and needs to be chosen more carefully to obtain better approximation ratios.

\begin{algorithm}[!h]
\setstretch{1.5}
\caption{Local ratio product state algorithm for Transverse Prize Collecting Vertex Cover}
\label{alg:SPCVC}
\vspace{0.1in}
{\bf Given}: Graph $G(V,E)$ and Hamiltonian $H = \sum_i c_i \phi_i + \sum_{ij \in E} c_{ij} \frac{(I+Z_i)\otimes(I+Z_j)}{4} $.\\
{\bf Assumptions}: $\phi_i$ is a rank-1 projector, $c_i \geq 0$, $\mathrm{Tr}[Z_i \phi_i] \leq 0\, \forall i \in V$ and $c_{ij} \geq 0\, \forall\, ij \in E$.
\begin{algorithmic}[1]
\State For any $i \in V$ with $c_i =0$, set $\rho_i = \ketbra{1}{1}$
\State Choose an edge $ij \in E$ such that $\min\{c_i, c_j, c_{ij}\} > 0$, and let $H_{ij} = \phi_i + \phi_j + \lambda \frac{(I+Z_i)\otimes(I+Z_j)}{4}$ where $\lambda = \frac{2}{1-\text{Tr}(Z_i\phi_i)} + \frac{2}{1-\text{Tr}(Z_j\phi_j)}$
\State Update $H \rightarrow H - w_{ij}H_{ij}$ where $w_{ij} = \min\{c_i, c_j, \frac{c_{ij}}{\lambda}\}$

\State Repeat from step 1 until no more edges are left to choose in step 2.

\State For any remaining $i$ with $c_i > 0$, set $\rho_i = I-\phi_i$.
\end{algorithmic}
{\bf Output: } Product state $\rho = \otimes_i \rho_i$.
\vspace{0.075in}
\end{algorithm}

\LocalratioAlgForTPCVC*

\begin{proof}
    Let 
    \begin{align}
        H_{ij} = \phi_i + \phi_j + \lambda \frac{(I+Z_i)\otimes(I+Z_j)}{4} \text{ where } \lambda = \frac{2}{1-\text{Tr}(Z_i\phi_i)} + \frac{2}{1-\text{Tr}(Z_j\phi_j)}.
    \end{align}
    Since Transverse Prize Collecting Vertex Cover has no constraints, by theorem \ref{thm:local_ratio} it is enough upper bound the ratio of the maximum expected value from our ansatz to the smallest eigenvalue of $H_{ij}$, $\mu_{ij}^*$.
    
    At the end of the algorithm, every qubit-$i$ has two possible outcomes $I-\phi_i$ or $\ketbra{1}{1}$, and every pair of qubits $(i,j)$ has four possible outcomes. The expectation value $H_{ij}$ with respect to state $\ketbra{1}{1}\otimes \ketbra{1}{1}$ is greater than or equal to $\ketbra{1}{1}\otimes (I-\phi_j)$ and  $(I-\phi_i)\otimes \ketbra{1}{1}$. So we only need to consider $(I-\phi_i)\otimes (I-\phi_j)$ and $\ketbra{1}{1}\otimes \ketbra{1}{1}$ to calculate the maximum expected value for $H_{ij}$. But $\lambda$ is chosen such that the expectation value of $H_{ij}$ with respect to both the states is equal to each other which is
    \begin{align}
        \mathrm{Tr}[H_{ij}(I-\phi_i)\otimes (I-\phi_j)] = \mathrm{Tr}[H'\ketbra{1}{1}\otimes\ketbra{1}{1}] = 1 - \frac{(\mathrm{Tr}[Z_i\phi_i] + \mathrm{Tr}[Z_j\phi_j])}{2}.
    \end{align}
    The ratio of $1 - \frac{(\mathrm{Tr}[Z_i\phi_i] + \mathrm{Tr}[Z_j\phi_j])}{2}$ to the minimum eigenvalue of $H'$ is a monotonically increasing function of  $\mathrm{Tr}[Z_i\phi_i]$ and $\mathrm{Tr}[Z_j\phi_j]$ and therefore is maximum when $\mathrm{Tr}[Z_i\phi_i] = \mathrm{Tr}[Z_j\phi_j] = 0$ and is equal to $\approx 4.19387.. < 4.194$ obtained numerically.
\end{proof}


\section{EVC with entangled constraints is in $P$}
\label{sec:entangled_VC}

Recall from the constrained local Hamiltonian problem~\eqref{eq:stoq_vc} in \Cref{sec:tvc_complexity} that both classical and Transverse Vertex Cover have the constraints:
\begin{equation}
\tr[\ketbra{00}{00}_{ij} \rho]=0\,\,\,\,\, \forall ij\in E.
\end{equation}

We can further generalize these constraints by allowing arbitrary rank-$1$ constraints $\ket{\psi} \bra{\psi}_{ij}$, instead of $\ket{00}\bra{00}_{ij}$, where $\ket{\psi}$ is $2$-qubit state.  We may also allow for generic $1$-local objectives without restriction. Like both classical and Transverse Vertex Cover, we will further assume that $\ket{\psi}\bra{\psi}$ is SWAP invariant, to capture the input graph is undirected in those problems.  
\begin{definition}[\cqvc $(\{\phi_i\}, G, \ket{\psi})$]
    Given a set of $1$-local terms $\{\phi_i\}_{i=1}^n$, a graph $G=([n], E)$, and a (non-zero) $2-$qubit quantum state $\ket{\psi}\in (\mathbb{C}^2)^{\otimes 2}$ satisfying $SWAP\ket{\psi}\bra{\psi} (SWAP)^\dagger = \ket{\psi}\bra{\psi}$ define
\begin{align}\label{eq:1}
EVC(\{\phi_i\}, G, \ket{\psi}):=\,\,\,\,\,\,\,\,\,\,\,\,\,\,\,\,\,\,\,\,\,\, &\min \sum_i Tr[\phi_i \rho]\\
\label{eq:cqvc_const} s.t.\,\,\,\,\,\,\,\,\,\,\, &Tr[\ket{\psi}\bra{\psi}_{ij}\rho]=0 \,\,\,\,\, \forall ij\in E,\\
 &Tr(\rho)=1,\\
 &\rho \semigeq 0.
\end{align}
\end{definition}

The assumed SWAP invariance places strong restrictions on the form of $\ket{\psi}$ up to local unitary.  Elementary considerations force $\ket{\psi}$ to be of a simple form. The following characterization is well-known and implicit in several works, we provide a proof for completeness.
\begin{proposition}[\cite{horodecki1996information}]\label{prop:const_class}
    $EVC(\{\phi_i\}, G, \ket{\psi})$ is equivalent to $EVC(\{\phi_i'\}, G, \ket{\psi'})$ where $\ket{\psi'}=\alpha\ket{00}+\beta \ket{11}$ for $\alpha, \beta$ non-negative real numbers or $\ket{\psi'}=\ket{01}- \ket{10}$.
\end{proposition}
\begin{proof}
    First note that, since $SWAP$ is Hermitian, it must be that $SWAP\ket{\psi}=\pm \ket{\psi}$. If $SWAP \ket{\psi}=-\ket{\psi}$ then it must be that $\ket{\psi}=\ket{01}-\ket{10}$ as this is the only vector in the antisymmetric subspace for $2$ qubits (up to normalization and phase).  Otherwise we may assume $\ket{\psi}$ is in the symmetric subspace and can be written as $\ket{\psi}=\alpha \ket{00} +\beta (\ket{01}+\ket{10}) +\gamma \ket{11}$.  
    Note that for any single qubit unitary $U$ we may take $\ket{\psi'}=U\otimes U \ket{\psi}$ and obtain an equivalent problem simply by rotating the objective function accordingly (if $\ket{\psi} \rightarrow U^{\otimes 2}\ket{\psi}$ then $\phi_i\rightarrow U^\dagger \phi_i U$ for all $i$).  We also note that the action of $U$ on $\ket{\psi}$ can be simply described as follows.
    Define 
    \begin{equation}
        M=\begin{bmatrix}
            \alpha & \beta \\
            \beta & \gamma
        \end{bmatrix}.
    \end{equation}
    Let $U\in \mathbb{C}^{2\times 2}$ be an arbitrary unitary matrix and let $\ket{\psi'}:= U^{\otimes 2} \ket{\psi}=\alpha' \ket{00}+\beta'(\ket{01}+\ket{10})+\gamma' \ket{11}$.  It is easily verified that 
    \begin{equation}
        \begin{bmatrix}
            \alpha' & \beta'\\
            \beta' & \gamma'
        \end{bmatrix}
        =U M U^T.
    \end{equation}
    $M$ is complex symmetric (not Hermitian) and so by the Takagi decomposition \cite{horn2012matrix}, we may find a unitary matrix $U$ such that $U M U^T$ is a diagonal matrix with non-negative entries.  It follows that we may find a unitary such that $U^{\otimes 2}\ket{\psi}=\alpha \ket{00}+\beta \ket{11}$ and so by rotating the objective we can obtain an equivalent problem with $\ket{\psi}=\alpha\ket{00}+\beta\ket{11}$.  
\end{proof}

Counter-intuitively, while \cqvc{} appears to be a well-motivated ``quantum'' generalization of vertex cover, it is polynomial time computable on a classical computer whenever the constraint is entangled.  As stated in the introduction this fact is implicit in previous works \cite{de2010ground, ji2011complete}, we provide a simplified proof relevant to \cqvc{}.

\begin{theorem}[\cite{de2010ground, ji2011complete}]\label{thm:cvqc}
    If $\ket{\psi}$ is not local unitary equivalent to $\ket{00}$ then $EVC(\{\phi_i\}, G, \ket{\psi})$ can be evaluated in polynomial time.  
\end{theorem}

\Cref{prop:const_class} implies that up to unitary rotation the constraint state $\ket{\psi}$ can take only two cases corresponding to two forms that $\ket{\psi}$ can take. Also, as we will show, there are only two cases for the graph $G$ which need to be studied to determine the complexity of \cqvc{}.  If $G$ is connected and bipartite then any other $G$ which is bipartite with the same partition of qubits will have the same feasible space.  In other words, the actual edges between the components of the partition do not affect the complexity of the feasible region as long as the graph is connected.  Similarly if the graph is not bipartite then the complexity of the feasible region of \cqvc{} will be the same as any other instance with the same number of qubits.  So the proof falls into $4$ cases and we demonstrate that \cqvc{} is solvable in each of these cases in turn.


\subsection{Proof of \cref{thm:cvqc}}

We begin with a few simple observations.  First note that the optimal $\rho$ can be assumed pure WLOG.  Indeed if $\rho=\sum_k \lambda_k \ket{\alpha_k}\bra{\alpha_k}$ then $0=Tr[\ket{\psi}\bra{\psi}_{ij} \rho]=\sum_k \lambda_k \bra{\alpha_k} \cdot \ket{\psi}\bra{\psi}_{ij} \cdot \ket{\alpha_k}$.  Each term in the sum is non-negative so each term must be zero.  The objective is $Tr[\sum_i \phi_i \rho]=\sum_k \lambda_k \bra{\alpha_k} \sum_i \phi_i \ket{\alpha_k}$.  The objective achieved by $\rho$ is a convex combination of the objectives achieved by pure states, hence at least one of them must have objective matching $\rho$.  We can also assume the graph $G$ is connected since if the graph is not connected we can solve \cqvc{} on each connected component and add up the optimal objectives.  

Now we turn to analyzing the constraints.  Let $\ket{\gamma}$ be an optimal state and express it in the computational basis:
\begin{equation}
    \ket{\gamma}=\sum_{\alpha_1, ..., \alpha_n=0}^1 \gamma_{\alpha_1, ..., \alpha_n} \ket{\alpha_1, ..., \alpha_n}.
\end{equation}
$\gamma_{\alpha_1, ..., \alpha_n}$ is an n-index tensor and constraints from \Cref{eq:cqvc_const} translate to tensor contractions:
\begin{equation}
    0=Tr[\ket{\psi}\bra{\psi}_{12} \ket{\gamma} \bra{\gamma}] \Leftrightarrow \sum_{\alpha_1, \alpha_2} \psi_{\alpha_1, \alpha_2}\gamma_{\alpha_1, \alpha_2, ..., \alpha_n}=0 \,\,\,\,\,\,\, \forall \alpha_3, ..., \alpha_n.
\end{equation}
In the remainder of the paper we will use Einstein summation convention and not explicitly include the sum.  The proof strategy will be to use the constraints from \Cref{eq:cqvc_const} to derive new $2-$local constraints not contained in the problem description.  For this we will employ the transfer matrix method pioneered by Bravyi:  
\begin{lemma}[\cite{bravyi2011efficient}]\label{lem:bravyi}
Let $\gamma_{\alpha_1, \alpha_2, ..., \alpha_n}$ be some tensor corresponding to a quantum state and let $\phi_{\alpha_1, \alpha_2}$ and $\theta_{\alpha_2, \alpha_3}$ be two tensors which satisfy:
\begin{equation}
\phi_{\alpha_1, \alpha_2} \gamma_{\alpha_1, ..., \alpha_n}=0 \,\,\,\,\,\,\,\, \theta_{\alpha_2, \alpha_3} \gamma_{\alpha_1, ..., \alpha_n}=0
\end{equation}
Then if we define $\omega_{\alpha_1, \alpha_3}=\phi_{\alpha_1, \beta}\epsilon_{\beta, \gamma}\theta_{\gamma, \alpha_3}$ it holds that $\omega_{\alpha_1, \alpha_3}\gamma_{\alpha_1, ..., \alpha_n}=0$.
\end{lemma}
The precise statement we need is a simple corollary of this fact.  This corollary has been used in other works, and our case of interest is a special case of the result established there.
\begin{corollary}[\cite{ji2011complete, de2010ground}]\label{cor:bravyi}
    Let $\gamma_{\alpha_1, ..., \alpha_n}$ be an n-index tensor, let $\ket{\epsilon}=\ket{01}-\ket{10}$ and let $\ket{\psi}=\alpha\ket{00}+\beta\ket{11}$ with $\alpha\neq 0 \neq \beta$.
    \begin{enumerate}
        \item \label{eq:const_plus_const} If $\psi_{\alpha_i \alpha_j}\gamma_{\alpha_1, ..., \alpha_n}=0$ and $\psi_{\alpha_j \alpha_k}\gamma_{\alpha_1, ..., \alpha_n}=0$ then $\epsilon_{\alpha_i \alpha_k} \gamma_{\alpha_1, ..., \alpha_n}=0$.
        \item \label{eq:const_plus_eps} If $\epsilon_{\alpha_i \alpha_j}\gamma_{\alpha_1, ..., \alpha_n}=0$ and $\psi_{\alpha_j \alpha_k}\gamma_{\alpha_1, ..., \alpha_n}=0$ then $\psi_{\alpha_i \alpha_k} \gamma_{\alpha_1, ..., \alpha_n}=0$.
        \item \label{eq:eps_plus_eps} If $\epsilon_{\alpha_i \alpha_j}\gamma_{\alpha_1, ..., \alpha_n}=0$ and $\epsilon_{\alpha_j \alpha_k}\gamma_{\alpha_1, ..., \alpha_n}=0$ then $\epsilon_{\alpha_i \alpha_k} \gamma_{\alpha_1, ..., \alpha_n}=0$.
    \end{enumerate}
\end{corollary}
\begin{proof}
    The following calculations may be easily verified:
    \begin{align}
        \psi_{\delta \zeta}\epsilon_{\zeta \eta}\psi_{\eta, \mu} \propto \epsilon_{\delta \mu},\\
        \epsilon_{\delta \zeta}\epsilon_{\zeta \eta}\psi_{\eta \mu}\propto \psi_{\delta \mu},\\
        \epsilon_{\delta\zeta}\epsilon_{\zeta \eta}\epsilon_{\eta \mu }\propto \epsilon_{\delta \mu},
    \end{align}
    with nonzero constant of proportionality.  \Cref{lem:bravyi} then implies each of the items.  
\end{proof}

The above corollary shows that if e.g. an optimal state $\ket{\gamma}$ for $EVC$ satisfies the constraint $\psi$ along the edge $(1, 2)$ and along the edge $(2, 3)$ then we may assume that it also satisfies the constraint $\epsilon$ along the edge $(1, 3)$.  We will say that $i$ and $j$ are connected by a $\psi$ edge if $\psi_{\alpha_i \alpha_j}\gamma_{\alpha_1, ..., \alpha_n}$ and we will say $i$ and $j$ are connected by an $\epsilon$ edge if $\epsilon_{\alpha_i, \alpha_j}\gamma_{\alpha_1, ..., \alpha_n}$.  Obviously for all $ij\in E$ in the original graph $i$ and $j$ are connected by a $\psi$ edge, and we may combine $\psi$ constraints on overlapping edges to derive new constraints using \Cref{cor:bravyi} as in \cite{de2010ground}.
The proof will depend on the constraint state $\ket{\psi}$ as well as the bipartiteness of $G$.  The form of $\ket{\psi}$ is restricted by \Cref{prop:const_class} so there are only a few cases to resolve.

\noindent \textbf{Case 1: Bipartite, $\ket{\psi}=\alpha\ket{00} +\beta \ket{11}$}

Assume the graph is connected and bipartite.  Let $(A, B)$ be the partition of the vertices such that $ij\in E$ implies exactly one of $\{i, j\}$ is in $A$ and exactly one is in $B$.  
Let $\ket{\gamma}$ be the optimal solution.  First we will demonstrate that for any two distinct points $i, j\in A$ we can assume that $\epsilon_{\alpha_i \alpha_j} \gamma_{\alpha_1, \alpha_2, ..., \alpha_n}=0$.  Since the graph is connected there is a path in $G$ $(i, k_1, k_2, ..., k_{2p-1}, j) $ such that vertices $k_\ell$ with $\ell$ odd are in $B$ and $k_\ell$ with $\ell$ even are in $A$. By \cref{cor:bravyi}, \cref{eq:const_plus_const} we can assume all the points in the path of distance $2$ which and in $A$ are connected with and $\epsilon$ edge: $\psi_{k_{ 2i} k_{2i+1}}\gamma_{\alpha_1, ..., \alpha_n}=0$ and $\psi_{k_{ 2i+1} k_{2i+2}}\gamma_{\alpha_1, ..., \alpha_n}=0$ implies $\epsilon_{k_{ 2i} k_{2i+2}}\gamma_{\alpha_1, ..., \alpha_n}=0$.  Then by \Cref{cor:bravyi}, \cref{eq:eps_plus_eps} we may use the newly derived constraints to connect all the points of distance $4$: $\epsilon_{k_{2i} k_{2i+2}}\gamma_{\alpha_1, ..., \alpha_n}=0$ and $\epsilon_{k_{2i+2} k_{2i+4}}\gamma_{\alpha_1, ..., \alpha_n}=0$ imply that $\epsilon_{k_{2i} k_{2i+4}}\gamma_{\alpha_1, ..., \alpha_n}=0$.  Similarly we may also connect points in $A$ of distance $6$ with an $\epsilon$ edge.  Proceeding in this way we may eventually connect $i$ and $j$ with an $\epsilon$ edge as well as any pair of points in $A$ and any pair of points in $B$.  Hence we will assume that all points in $A$ are connected with an $\epsilon $ edge as well as all points in $B$.  

The next observation is that we can assume any $i\in A$ any $j\in B$ are connected with a $\psi $ edge.  Fix $i\in A$ and $j\in B$.  $i$ must be connected to some other point $k\in B$ since the graph is connected.  $k$ must be connected to $j$ with an $\epsilon$ edge.  \Cref{cor:bravyi} \Cref{eq:const_plus_eps} then implies that $i$ and $j$ are connected by a $\psi$ edge.  Now we can use the tensor constraints we have described to fix the structure of a feasible state.  Let us write $\ket{\gamma}$ in the computational basis by first listing the indices in $A$ then listing the indices in $B$
\begin{equation}
\ket{\gamma}=\sum_{\substack{\mathbf{x}\in \mathbb{F}_2^{|A|}\\\mathbf{y}\in \mathbb{F}_2^{|B|}}}\gamma_{\mathbf{x}, \mathbf{y}} \ket{\mathbf{x}, \mathbf{y}}
\end{equation}
Observe that the epsilon edges force:
\begin{align*}
\epsilon_{\alpha_1, \alpha_2} \gamma_{\alpha_1, \alpha_2, \alpha_3 ...}=0 \,\,\,\,\,\,\,\,\forall \alpha_3, ..., \alpha_n  \Leftrightarrow  \gamma_{0, 1, \alpha_3 ...}-\gamma_{1, 0, \alpha_3 ...}=0 \,\,\,\,\,\,\,\,\forall \alpha_3, ..., \alpha_n\\
\Leftrightarrow \gamma_{0, 1, \alpha_3 ...}=\gamma_{1, 0, \alpha_3 ...} \,\,\,\,\,\,\,\,\forall \alpha_3, ..., \alpha_n.
\end{align*}
Hence we may permute the entries of $\mathbf{x}$ and $\mathbf{y}$ separately without changing the amplitude.  Let us rewrite $\ket{\gamma}$ as 
\begin{equation}\label{eq:hamm_decomp}
    \ket{\gamma}=\sum_{a=0}^{|A|} \sum_{b=0}^{|B|} \gamma_{a, b} \ket{a, b},
\end{equation}
where 
\begin{equation}
    \ket{a, b}=\frac{1}{\sqrt{\binom{|A|}{a} \binom{|B|}{b}} }\sum_{\substack{\mathbf{x}:|\mathbf{x}|=a\\\mathbf{y}: |\mathbf{y}|=b}} \ket{\mathbf{x}, \mathbf{y}}.
\end{equation}

We can also relate different entries of $\gamma$ using $\psi$ edges.  Let $\gamma$ satisfy $\psi_{ij}$ and let $\mathbf{x}$ and $\mathbf{y}$ be vectors such that $\mathbf{x}_i=\mathbf{y}_j=0$.  Let $\mathbf{x}'$ and $\mathbf{y}'$ be vectors which are identical to $\mathbf{x}$ and $\mathbf{y}$ respectively, except with $\mathbf{x}_i'=1=\mathbf{y}_j'$.  Then, 
\begin{equation} 
\alpha \gamma_{\mathbf{x}, \mathbf{y}}+\beta \gamma_{\mathbf{x}', \mathbf{y}'}=0 \Leftrightarrow \gamma_{\mathbf{x}', \mathbf{y}'}=\frac{-\alpha}{\beta} \gamma_{\mathbf{x}, \mathbf{y}}.
\end{equation}
It follows that 
$$
\gamma_{a, b}=\gamma_{a-1, b-1} \frac{-\alpha}{\beta} \text{ if $a-1\geq 0$ and $b-1 \geq 0$.}
$$
Hence we can further rewrite $\ket{\gamma}$ as 
\begin{equation}\label{eq:final_decomp}
    \ket{\gamma}=\gamma_0 \ket{0}+\sum_{a=1}^{|A|} \gamma_a \ket{a} +\sum_{b=1}^{|B|}\gamma_b \ket{b},
\end{equation}
where $\ket{a}, \ket{b}, \ket{0}$ are orthonormal vectors satisfying:
\begin{equation}
    \ket{a} \propto \sum_{j=0}^{\min(|A|-a, |B|)} \left( \frac{-\alpha}{\beta}\right)^j \ket{a+j, j},
\end{equation}
\begin{equation}
    \ket{b} \propto \sum_{j=0}^{\min(|B|-j, |A|)} \left( \frac{-\alpha}{\beta}\right)^j \ket{j, b+j},
\end{equation}
and 
\begin{equation}
    \ket{0}\propto\sum_{j=0}^{\min(|A|, |B|)} \left( \frac{-\alpha}{\beta}\right)^j \ket{j, j}.
\end{equation}
Note that the normalization constants are easily calculated since $\ket{a}, \ket{b}$ are themselves polynomially large sums of orthonormal states.

We have demonstrated that the constraints force $\gamma$ to lie in a subspace of dimension equal to the total number of qubits.  Our next step is to show that for such states $\ket{\gamma}$ that we can calculate the objective of $EVC$ and then that we can optimize to find the optimal state.  It will be convenient to ``revert'' to the older decomposition of $\ket{\gamma}$ in \Cref{eq:hamm_decomp} since it will have applications for the other proof cases.  Note that an algorithm which computes the objective in the decomposition of \Cref{eq:hamm_decomp} can also be used to compute the objective in the decomposition in \Cref{eq:final_decomp}.  Let us assume $i\in A$ and compute the matrix elements of $\ket{0}\bra{0}_i$ in this decomposition. 

\begin{align}
    \bra{a, b} \ket{0} \bra{0}_i \ket{a', b'}=\frac{1}{\sqrt{\binom{|A|}{a}\binom{|B|}{b}\binom{|A|}{a'}\binom{|B|}{b'}}}\sum_{\mathbf{x}, \mathbf{y}, \mathbf{x}', \mathbf{y}'} \bra{\mathbf{x}, \mathbf{y}} \ket{0} \bra{0}_i \ket{\mathbf{x}', \mathbf{y}'},
\end{align}
where the sum is over $(\mathbf{x}, \mathbf{y}, \mathbf{x}', \mathbf{y}')$ satisfying $|\mathbf{x}|=a$, $|\mathbf{y}|=b$, $|\mathbf{x}'|=a'$ and $|\mathbf{y}'|=b'$.  Note that it must be $a=a'$ and $b=b'$ for the overall sum to be nonzero.  Also observe that for a nonzero value inside the sum we must have $\mathbf{y}=\mathbf{y}'$, $\mathbf{x}_j=\mathbf{x}'_j$ for $j\neq i$ and $\mathbf{x}_i =\mathbf{x}_i'=0$.  Hence we can easily calculate the value of the sum to be 
\begin{align}\label{eq:body_00_calc}
    =\delta_{a, a'}\delta_{b, b'} \frac{1}{\binom{|A|}{a}\binom{|B|}{b}} \cdot \binom{|B|}{b} \binom{|A|-1}{a}=\delta_{a, a'}\delta_{b, b'} \frac{|A|-a}{|A|} .
\end{align}
Computation of the values of $\{\ket{0}\bra{1}_i, \ket{1}\bra{0}_i, \ket{1}\bra{1}_i\}$ are similar and hence are relegated to the appendix.  Since we can write each $\phi_i$ as a sum of these operators, we can compute the matrix elements of $\sum_i \phi_i$ in the decompsition \Cref{eq:hamm_decomp} and hence in the decomposition \Cref{eq:final_decomp}.  

Let us define a Hermitian matrix $M$ with rows indexed by the union $A\cup B$ defined as:
\begin{align}
    M_{a, a'} =\bra{a} \sum_i \phi_i \ket{a'} \text{ if $a,a'\in A$},\\
    M_{a, b} =\bra{a} \sum_i \phi_i \ket{b} \text{ if $a\in A, b\in B$},\\
    M_{b, b'} =\bra{b} \sum_i \phi_i \ket{b'} \text{ if $b,b'\in B$}.
\end{align}
With these calculations in hand, EVC reduces to computing the smallest eigenvalue of $M$:
\begin{align}
    \min \bra{\gamma} \sum_i \phi_i \ket{\gamma}\\
    =\min \sum_{a, a'\in A} \gamma_a^* \gamma_a  \bra{a} \sum_i \phi_i \ket{a'}+ \sum_{a\in A, b\in B} \gamma_a^* \gamma_b\bra{a} \sum_i \phi_i \ket{b}\\
    +\sum_{a\in A, b\in B} \gamma_a \gamma_b^*\bra{b} \sum_i \phi_i \ket{a} +\sum_{b, b'\in B}\gamma_b^* \gamma_{b'} \bra{b} \sum_i \phi_i \ket{b'}\\
    =\lambda_{min} (M).
\end{align}

\noindent \textbf{Case 2: $\ket{\psi}=\ket{01}-\ket{10}$}

If the graph is bipartite then we can repeat the proof for the $\ket{\psi}=\alpha \ket{00} +\beta\ket{11}$ while replacing $\psi$ with $\epsilon$ to see that every pair of vertices is connected with an $\epsilon $ edge.  Hence, just as in the proof for that case, if we write $\ket{\gamma}=\sum_{\mathbf{x}\in \mathbb{F}_2^{n}}\gamma_{\mathbf{x}} \ket{\mathbf{x}}$ then we can infer $\gamma_{\mathbf{x}}=\gamma_{\mathbf{x}'}$ for any pair with $|\mathbf{x}|=|\mathbf{x}'|$.  So $\ket{\gamma}=\sum_{a=0}^n\gamma_a \ket{a}$ where 
\begin{equation}\label{eq:non_bi_hamm}
    \ket{a}\propto \sum_{\mathbf{x}: |\mathbf{x}|=a} \ket{\mathbf{x}}.
\end{equation}  
We can compute the matrix elements $\bra{a} \sum_i \phi_i \ket{a}$ for this case in the same way (see the appendix) so once again EVC reduces to an eigenvalue problem.  If the graph is not bipartite then we can find a bipartite subgraph and use it to connect every pair of vertices with an $\epsilon$ edge so this case reduces to the bipartite case.  

\noindent \textbf{Case 3: Non-bipartite, $\ket{\psi}=\alpha\ket{00}+\beta \ket{11}$}

We will apply the same kind of analysis for this case to show that the case where there is an odd cycle in the graph is even more restrictive and that it can be solved by finding the smallest eigenvalue of a $2\times 2 $ matrix.  Find a bipartite subgraph of $G$ with parititon $(A, B)$ and repeat the argument for the bipartite case.  We can assume then that every pair of vertices in $A$ is connected with an $\epsilon$ edge,every pair in $B$ is connected with an $\epsilon $ edge and every pair $i\in A, j\in B$ is connected with a $\psi$ edge.  Since the original graph was not bipartite there must be a $\psi$ edge joining a pair in $A$ or $B$.  Say WLOG (by possibly renaming the partitions) that there is a $\psi$ edge between vertices $i, j\in A$.  Since all vertices in $A$ are connected with $\epsilon $ edges then using \Cref{cor:bravyi} \Cref{eq:const_plus_eps} to connect $i$ and $j$ with any other point in $A$ with a $\psi$ edge.  So for any other pair of vertices in $A$, say $k$ and $\ell$ we may assume $i$ is connected to $\ell$ with a $\psi$ edge and $i$ is connected to $k$ with an $\epsilon$.  It follows using \Cref{eq:const_plus_eps} that $k$ and $\ell$ must also be connected with a $\psi$ edge.  Next we establish that any pair of points $i\in A$ and $j\in B$ are connected with an $\epsilon $ edge.  Note that $j$ must be connected to some $k\in A, \neq i$ vertex with a $\psi$ edge so we may use \Cref{eq:const_plus_const} to connect $i$ and $j$ with an $\epsilon$ edge.  Finally we will note that we can connect all vertices in $B$ with a $\psi$ edge.  Indeed, for any pair $i, j\in B$ for any point $k\in A$ we know that $i$ is connected with $k$ using an $\epsilon$ edge and $j$ is connected to $k$ using a $\psi$ edge.  \Cref{eq:const_plus_eps} implies that $i$ and $j$ may be connected with a $\psi$ edge.

We have demonstrated for this case that we may assume any two points in the graph are connected with a $\psi$ edge and an $\epsilon$ edge.  So, just as the previous cases we can assume $\gamma_{\mathbf{x}}=\gamma_{\mathbf{x}'}$ whenever $|\mathbf{x}|=|\mathbf{x}'|$.  $\ket{\gamma}=\sum_{a=0}^{n}\gamma_a \ket{a}$ where the $\psi$ constraints imply $\gamma_a=\gamma_{a-2}(-\alpha/\beta)$.  Since the ``even'' and ``odd'' amplitudes are related we may write:
\begin{equation}
\ket{\gamma}=\gamma_{e} \ket{e}+\gamma_o \ket{o},    
\end{equation}
where
\begin{equation}
    \ket{e}\propto \sum_{a \,\, even, \,\,\leq n+1} \left(\frac{-\alpha}{\beta}\right)^{a/2} \ket{a}
\end{equation}
and
\begin{equation}
    \ket{o}\propto \sum_{a \,\, odd, \,\,\leq n+1} \left(\frac{-\alpha}{\beta}\right)^{(a-1)/2} \ket{a}.
\end{equation}

By case $2$ we can compute $\bra{a} \phi_i \ket{b}$ where $\ket{a}$ and $\ket{b}$ are defined in \Cref{eq:non_bi_hamm} so we can compute $\bra{e} \phi_i \ket{o}$, $\bra{e} \phi_i \ket{e}$, etc. and the problem once again reduces to an eigenvalue problem.



\section*{Acknowledgements}
    C.R. would like to thank Jun Takahashi and Andrew Zhao for insightful discussions during this work. Sandia National Laboratories is a multimission laboratory managed and operated by National Technology and Engineering Solutions of Sandia, LLC., a wholly
owned subsidiary of Honeywell International, Inc., for the U.S. Department of Energy’s National
Nuclear Security Administration under contract DE-NA-0003525. This work was supported by
the U.S. Department of Energy, Office of Science, Office of Advanced Scientific Computing Research, Accelerated Research in Quantum Computing. O.P. was also supported by U.S. Department of Energy, Office of Science, National Quantum Information Science Research Centers.

\printbibliography
\begin{appendices}

    \section{Additional calculations for Proof of \Cref{thm:cvqc}} \label{prf:EVC_2}

In this section we will use the notation $\delta\bigg[ P\bigg]$ as the delta function which evaluates to $1$ if the predicate $P$ is true and $0$ otherwise.  

    \noindent \textbf{Bipartite Case}
    
    Suppose we have $n$ qubits and let $(A, B)$ be a partition of the qubits.  For all $a\leq |A|$  and $b\leq |B|$ define the following states
    \begin{align}
        \ket{a, b}=\frac{1}{\sqrt{\binom{|A|}{a}\binom{|B|}{b}}} \sum_{\substack{\mathbf{x}\in \mathbb{F}_2^{|A|} |\mathbf{x}|=a\\\mathbf{y}\in \mathbb{F}_2^{|B|}: |\mathbf{y}|=b}} \ket{\mathbf{x}, \mathbf{y}},
    \end{align}
    where $\ket{\mathbf{x}, \mathbf{y}}$ is a computational basis state in which the qubits the state of the qubits from $A$ is given in $\mathbf{x}$ and the state of the qubits in $B$ is given in $\mathbf{y}$.  For all $i \in A$ and $j\in B$ we can derive the following:
    \begin{align}
    \bra{a, b} \ket{0} \bra{1}_i \ket{a', b'}=\frac{1}{\sqrt{\binom{|A|}{a}\binom{|B|}{b}\binom{|A|}{a'}\binom{|B|}{b'}}}\sum_{\substack{|\mathbf{x}|=a\\|\mathbf{x}'|=a' \\ |\mathbf{y}|=b \\ |\mathbf{y}'|=b'}} \bra{\mathbf{x}, \mathbf{y}} \ket{0} \bra{1}_i \ket{\mathbf{x}', \mathbf{y}'}\\
    \nonumber = \frac{1}{\sqrt{\binom{|A|}{a}\binom{|B|}{b}\binom{|A|}{a'}\binom{|B|}{b'}}}\sum_{\substack{|\mathbf{x}|=a\\|\mathbf{x}'|=a' \\ |\mathbf{y}|=b \\ |\mathbf{y}'|=b'}} \delta\bigg[\mathbf{y}=\mathbf{y}'\bigg]\,\, \delta\bigg[\mathbf{x}_j=\mathbf{x}_j' \,\,\forall \,\,j\neq i, \mathbf{x}_i=0, \mathbf{x}_i'=1\bigg]\\
    \nonumber =\delta\bigg[ a+1=a', b=b'\bigg] \frac{\binom{|A|-1}{a} \binom{|B|}{b}}{\sqrt{\binom{|A|}{a}\binom{|A|}{a+1}}\binom{|B|}{b}} =\delta\bigg[ a+1=a', b=b'\bigg]  \frac{\sqrt{(a+1)(|A|-a)}}{|A|}.
    \end{align}
    \begin{align}
        \bra{a, b} \ket{1} \bra{1}_i \ket{a', b'}= \bra{a, b} (\mathbb{I}-\ket{0}\bra{0}_i) \ket{a', b'} =\delta\bigg[a=a', b=b' \bigg]\left( 1-\frac{|A|-a}{|A|}\right), 
    \end{align}
    where the last equality follows from \Cref{eq:body_00_calc}.  
    \begin{align}
        \bra{a', b'} \ket{1}\bra{0}_i \ket{a, b}=(\bra{a, b} \ket{0} \bra{1}_i \ket{a', b'})^\dagger.
    \end{align}
    Computing the analogous quantities of $j$ can be done by switching the role of $A$ and $B$.  The proof depends only on the values being explicit so we skip these calculations here.  

    \noindent \textbf{Non-bipartite Case}

    Suppose again that there are $n$ qubits and for each $a\in [n]$ define:
    \begin{equation}
        \ket{a}=\frac{1}{\sqrt{\binom{n}{a}}} \sum_{\mathbf{x}\in \mathbb{F}_2^n:|\mathbf{x}|=a} \ket{\mathbf{x}}.
    \end{equation}
    For all $a, b\in [n]$ we may calculate the following quantities.  
    \begin{align}
        \bra{a} \ket{0}\bra{0}_i \ket{b}=\frac{1}{\sqrt{\binom{n}{a}\binom{n}{b}}}\sum_{|\mathbf{x}|=a, |\mathbf{y}|=b} \bra{\mathbf{x}} \ket{0} \bra{0}_i \ket{\mathbf{y}}=\frac{\delta\bigg[ a=b \bigg]}{\binom{n}{a}} \sum_{|\mathbf{x}|, |\mathbf{y}|=a} \delta\bigg[ \mathbf{x}_i=\mathbf{y}_i=0, \mathbf{x}_j=\mathbf{y}_j \,\, \forall \,\, j\neq i  \bigg]\\
        \nonumber =\delta\bigg[ a=b\bigg] \frac{\binom{n-1}{a}}{\binom{n}{a}}=\frac{n-a}{n}.
    \end{align}
    \begin{align}
        \bra{a} \ket{1}\bra{1}_i \ket{b} =\bra{a} (\mathbb{I}-\ket{0}\bra{0}_i) \ket{b}=\delta\bigg[a=b\bigg] \left(1-\frac{n-a}{n} \right)
    \end{align}
    \begin{align}
        \bra{a} \ket{0}\bra{1}_i \ket{b} =\frac{1}{\sqrt{\binom{n}{a}\binom{n}{b}}}\sum_{|\mathbf{x}|=a, |\mathbf{y}|=b} \bra{\mathbf{x}} \ket{0} \bra{1}_i \ket{\mathbf{y}}=\frac{\delta\bigg[ a+1=b\bigg]}{\sqrt{\binom{n}{a} \binom{n}{a+1}}} \sum_{|\mathbf{x}|, |\mathbf{y}|=a} \delta\bigg[ \mathbf{x}_i=0, \mathbf{y}_i=1, \mathbf{x}_j=\mathbf{y}_j \,\, \forall \,\, j \neq i\bigg] \\
        \nonumber = \delta\bigg[ a+1=ab\bigg]  \frac{\sqrt{(a+1)(n-a)}}{n}.
    \end{align}
    \begin{align}
         \bra{b} \ket{0}\bra{1}_i \ket{a} =  (\bra{a} \ket{0}\bra{1}_i \ket{b})^\dagger
    \end{align}

\section{Review of perturbative Bloch expansion}
\label{appx:bloch_expansion}

Let us review some of the key ideas of perturbative Bloch expansion that are necessary for us. For more details, we redirect the reader to \cite{jordan2008perturbative} and the references within.

Let $H_0$ be an unperturbed Hamiltonian whose ground state energy is zero, and $P_0$ be the projector on to ground subspace $H_0$ with dimension $d$. Let $\lambda V$ be the perturbation to the Hamiltonian $H_0$, making the total Hamiltonian $H = H_0 + \lambda V$. Let $\ket{\psi_1}.....\ket{\psi_d}$ be the low energy perturbed eigenvectors of $H$ with eigenvalues $E_1,....,E_d$. Let $\ket{\alpha_i} = P_0 \ket{\psi_i}$ and for sufficiently small $\lambda$, the set $\left\{\ket{\alpha_i}\right\}$ are linearly independent and span the ground subspace on $H_0$. Let us define a linear operator $\mathcal{U}$ such that
\begin{align}
    \mathcal{U} \ket{\alpha_i} = \ket{\psi_i}\quad \forall \, i = 1,2....d
\end{align}
and
\begin{align}
    \mathcal{U}\ket{\phi} = 0 \quad \forall \, \ket{\phi} \, \text{ such that }\,  P_0\ket{\phi} = 0.
\end{align}
Similarly, let us define $\mathcal{U}^{-1}$ such that
\begin{align}
    \mathcal{U}^{-1} \ket{\psi_i} = \ket{\alpha_i}\quad \forall \, i = 1,2....d
\end{align}
and
\begin{align}
    \mathcal{U}^{-1}\ket{\phi} = 0 \quad \forall \, \ket{\phi} \, \text{ such that }\,  P_0\ket{\phi} = 0.
\end{align}
The Bloch expansion is a perturbative series expansion of $\mathcal{U}$ of the following form
\begin{align} \label{eq:U_expansion}
    \mathcal{U} = P_0 + \sum_{m=1}^{\infty} \mathcal{U}^{(m)}
\end{align}
where
\begin{align}
    \mathcal{U^{m}} = \lambda^m \sum_{(m)} S^{l_1}VS^{l_2}...VS^{l_m}VP_0
\end{align}
and the summation is over non-negative tuples $(l_1, l_2,...,l_m)$ such that
\begin{align}
    l_1+l_2+...+l_m &= m \\
    l_1+l_2+...+l_p &\geq p \quad \forall \, p=1,2,...,m-1
\end{align}
and
\begin{align}
    S^{l} = \begin{cases}
        \ \  \frac{I-P_0}{(-H_0)^l} & \text{if}\quad l > 0\\
        \ \  -P_0 &\text{if}\quad l = 0.
    \end{cases}
\end{align}
With the assumption that the ground state energy of $H_0$ is zero, we get $\frac{I-P_0}{H_0} = \frac{1}{H_0}$ which implies
\begin{align}
    S^{l} = \begin{cases}
        \ \  (-H_0)^{-l} & \text{if}\quad l > 0\\
        \ \  -P_0 &\text{if}\quad l = 0.
    \end{cases}
\end{align}

Let $\mathcal{A} = \lambda P_0 V \mathcal{U}$. Note that $\ket{\alpha_1},....\ket{\alpha_d}$ are the right eigenvectors of $\mathcal{A}$ with eigenvalues $E_1,...,E_d$. The effective Hamiltonian $H_{\text{eff}} = \mathcal{U}\mathcal{A}\mathcal{U}^{-1}$, and for our purposes, it is sufficient to approximate $H_{\text{eff}}$ with $\mathcal{A}$. Using \Cref{eq:U_expansion}, we can derive the perturbative expansion for $\mathcal{A}$ as 
\begin{align}
    \mathcal{A} = \sum_{m=1}^{\infty} \mathcal{A}^{(m)}
\end{align}
where
\begin{align}
    \mathcal{A}^{(m)} = \lambda^m \sum_{(m-1)} P_0VS^{l_1}VS^{l_2}....VS^{l_{m-1}}VP_0
\end{align}
and the summation is over non-negative tuples $(l_1, l_2,...,l_{m-1})$ such that
\begin{align}
    l_1+l_2+...+l_{m-1} &= m-1 \\
    l_1+l_2+...+l_p &\geq p \quad \forall \, p=1,2,...,m-2
\end{align}
\end{appendices}
\end{document}